\documentclass{ieeeaccess}
\usepackage{cite}
\usepackage{amsmath,amssymb,amsfonts}
\usepackage{algorithmic}
\usepackage{graphicx}
\usepackage{textcomp}

\usepackage{verbatim}
\usepackage{kotex}
\usepackage[utf8]{inputenc}
\usepackage{subfigure}
\usepackage{dsfont}
\usepackage{amssymb}
\usepackage{algorithm}
\usepackage{cuted}
\usepackage{epstopdf}
\usepackage{soul,color}
\usepackage{setspace}
\usepackage{multirow}
\usepackage{setspace}
\usepackage{lipsum}
\usepackage{multicol}
\usepackage{float}

\usepackage[top=0.61in, bottom=0.61in, left=0.60in, right=0.60in]{geometry}

\usepackage{amsthm}
\newtheorem{lem}{Lemma}

\newtheorem{pro}{Proposition}

\newcommand\blfootnote[1]{%
  \begingroup
  \renewcommand\thefootnote{}\footnote{#1}%
  \addtocounter{footnote}{-1}%
  \endgroup
}

\def\BibTeX{{\rm B\kern-.05em{\sc i\kern-.025em b}\kern-.08em
    T\kern-.1667em\lower.7ex\hbox{E}\kern-.125emX}}
\begin{document}
\history{Received March 22, 2023, Accepted April 24, 2023, date of current version April 24, 2023.\\
This work has been submitted to the IEEE for possible publication. Copyright may be transferred without notice, after which this version may no longer be accessible.}
\doi{10.1109/ACCESS.2023.3270631}

\title{
Machine Learning for Relaying Topology: Optimization of IoT Networks with Energy Harvesting
}
\author{
	\uppercase{Kiseop Chung}\authorrefmark{1} \IEEEmembership{Member, IEEE} 
	and \uppercase{Jin-Taek Lim}\authorrefmark{2} \IEEEmembership{Member, IEEE}
}
\address{Agency for Defense Development, Daejeon 34186, Republic of Korea (e-mail: cks991030@snu.ac.kr\authorrefmark{1}, jtlim870708@gmail.com\authorrefmark{2})}
\tfootnote{This work was supported by the Agency for Defense Development Grant funded by the Korean Government(2023).}

\markboth
{Chung \headeretal: Machine Learning for Relaying Topology: Optimization of IoT Network with Energy Harvesting}
{Chung \headeretal: Machine Learning for Relaying Topology: Optimization of IoT Network with Energy Harvesting}

\corresp{Corresponding author: Jin-Taek Lim}

\begin{abstract}

In this paper, we examine Internet of Things (IoT) systems related to smart cities, smart factories, connected cars, etc.
To support such systems in a wide area with low power consumption, energy harvesting technology utilizing wireless charging infrastructure is necessary for the longevity of networks.
Considering that the position and amount of energy charged for each device could be unbalanced according to the distribution of nodes and energy sources, maximizing the minimum throughput among all nodes has become an NP-hard challenging issue.
To overcome this challenge, we propose a machine learning based relaying topology algorithm with a novel backward-pass rate assessment method to present proper learning direction and an iterative balancing time slot allocation algorithm which can utilize a node with sufficient energy as the relay.
To validate our proposed scheme, we conducted simulations on our established system model; thus, we confirm that the proposed scheme is stable and superior to conventional schemes.
\end{abstract}

\begin{keywords}
unsupervised learning, variational autoencoder, IoT network, TDMA system, energy harvesting, relay.
\end{keywords}

\titlepgskip=-15pt

\maketitle

\section{Introduction}
\label{sec:introduction}
\blfootnote{
The preprint of this article may be found in:\\ 
http://arxiv.org/abs/2301.08481
}
Internet of Things (IoT) technology will bring enormous innovations for societal and industrial systems in terms of improved efficiency, sustainability, and safety.
By exchanging many types of information such as traffic, energy usage, and environment data (e.g., temperature, humidity), IoT devices can create smart-cities, connected industries, connected vehicles and integrated health care services \cite{Shafique20}.
As these application services must have numerous wireless devices with scalability and a long lifespan, there is an emerging need for easy-to-maintain wireless networks with a simplified topology.
In addition, networks such as these do not require a high data rate for transmitted data, but require several different characteristics, such as frame sizes to be in the order of tens of bytes, intermittent transmission to be a few times per day, ultra-low speeds with an insensitive delay requirement, and mostly uplink-centric transmission.
These requirements for an uplink-centric, low data rate and maintenance free network result in wireless-powered star-of-stars network technologies where a new air interface provides an energy-efficient relaying alternative to available wireless network systems while covering larger areas.

Recently, a study on a relay-based star-of-stars topology for energy efficiency was conducted to cover a large area \cite{Li12}.
\cite{Li12} proposed a cooperative multi-hop transmission scheme in wireless sensor networks considering the circuit energy consumption, and proved that the energy consumption per unit transmit distance can be minimized.
As a new method for long-term operation without maintenance, research on energy harvesting-based IoT networks are also being actively conducted \cite{Ma19}.

In the energy harvesting-based IoT research area, there are some issues where the amount of energy charged by each node is different according to the distribution of the energy sources (e.g., power beacons (PBs)) and nodes, with a varying amount of residual/remaining energy in the devices.
\cite{Benkhelifa21} studied the resource allocation in energy harvesting (EH)-enabled Long Range (LoRa) networks with external energy sources such as power beacons.
They solved the problem of regional energy distribution by the sub-optimal spreading factor and power allocation algorithm.
Moreover, \cite{Leu14} proposed a clustering algorithm for grouping sensor nodes based on the energy distribution of the remaining nodes.
These studies maximize the energy efficiency of the network through scheduling under the distribution of energy.
Furthermore, \cite{He20} proposed a distributed protocol maximizing the minimum sensing rate of sources operating in energy harvesting-based IoT networks.
They suggested a method for determining nodes to be charged, transmission power and charging duration of a PB, and routing of data by nodes and link scheduling.

However, prior studies have yet to considered that, if the role of the node is not fixed, that is, if it is allowed to operate as a relay optimally adapting to the change of the energy distribution, energy efficiency could be improved.
More specifically, if a certain node is close to PBs, it may have more harvested energy.
Moreover, nodes closer to the sink have a lower required energy for transmission; conversely, nodes farther away from the sink have a higher required energy.
Accordingly, if some nodes with sufficient energy can transmit data on behalf of nodes with lower energy, the efficiency of the entire network could be maximized.
However, these kinds of routing problems are usually formulated as a max-min mixed integer optimization problem, which hinders solving the optimal solution in polynomial time due to its NP-hardness and non-linearity of the problem \cite{Lodi09}.

\subsection{Related Works}

Recently, deep learning based techniques have been applied to various wireless network problems in \cite{Lee19,Luo19,Lee21-1}.
Unlike traditional solutions based on mathematical models, deep learning provides the solution without the complex manipulation of mathematics.
The neural network itself is trained to produce the optimal solution based on the given data or learning direction.
Applying deep learning to existing communication networks can largely be classified into supervised learning or unsupervised learning, depending on whether there is a data set or not.
There has been one recent work on bringing network topology graph information into learning models in the communication community.

\cite{Mao17} proposed a network routing method based on supervised deep learning techniques which shows adequate results on simple topologies, but is less effective on complex ones.
In addition, the application of the deep learning for handling on graph-related problems was examined in \cite{Zhang20}, based on supervised learning.
In addition to topology problems, deep learning is being applied to many communication studies.
\cite{Lee21-2} designed an efficient deep neural network framework and a novel training strategy for wireless-powered secure communication with near-optimal performance and a low computation time.
As can be seen, the solution using supervised learning is achieving modest results, but it its premise is based on the possession of a vast labeled dataset for learning.

Unfortunately, since the routing topology problem defined above is NP-hard and non-linear, the labeled dataset for supervised learning is not easily obtainable.
Unsupervised learning does not require a dataset, making it generally important to design a loss function or rewards that can set the correct direction for learning.
Deep Reinforcement Learning (DRL), one of the unsupervised learning techniques, is about an agent interacting with the environment: learning an optimal policy by trial and error for sequential decision making problems.
In a given communication problem, if the environment, state, reward, and action are properly modeled, the solution for the problem can be found through DRL without a dataset.
Therefore, the optimization level of the solution depends on how accurately the environment, states, reward and action are modeled.
Usually, such modeling requires considerable skill as the problem becomes more complex.
Recently, \cite{Ding20} proposed an energy-efficient fair communication through trajectory design and band allocation (EEFC-TDBA) which allows an unmanned aerial vehicle (UAV) to adjust the flight speed and direction to enhance energy efficiency and allocate a frequency band to achieve fair communication service.
The DRL algorithm for solving the latency minimization problem for both communication and computation in a maritime UAV swarm mobile edge computing network was suggested and analyzed in \cite{Liu22}.

Additionally, the variational autoencoder (VAE), first proposed by \cite{Kingma13}, could be an unsupervised learning technique that can avoid the complex modeling of DRL and can solve the NP-hard and non-linear problem in an unsupervised manner \cite{Zhu22}.
VAEs are probabilistic generative models that require neural networks as the encoder and decoder for the first and second component, respectively.
The encoder maps the input variable to a latent space that corresponds to the parameters of a variational distribution.
The decoder has the opposite function: to map from the latent space to the input space to produce or generate data points.
Configured in this way, the encoder and decoder can generate multiple different samples that all come from the same latent distribution.
Therefore, the applications of VAEs are usually found in generating data for speech, images, and text.
For example, \cite{Simonovsky18} implemented a chemical molecule graph generative model using VAE, though its application was limited only to smaller graphs and required a dataset for training, making it difficult to apply directly to the NP-hard problem.

In the communication field, VAE has been applied in the following recent studies.
\cite{Lauinger22} analyzed the VAE in terms of performance and flexibility over a classical additive white Gaussian noise channel with inter-symbol interference and over a dispersive linear optical dual-polarization channel, showing that it can extend the application range of blind adaptive equalizers.
A probabilistic model based on VAEs was proposed for short packet wireless communication systems in \cite{Alawad22}, where the information messages are represented by the so called packet hot vectors which are inferred by the VAE latent random variables.

However, the generative capabilities of the VAE decoder can be used as a novel method to solve optimization problems.
After appropriately reflecting the optimization problem in the loss function for learning the decoder, the solution for the optimization problem can be obtained by applying various latent inputs during learning.
Therefore, in this paper, we apply the VAE to solve the NP-hard and non-linear IoT problem defined above.

\subsection{Contributions and Organization}

The main contributions of this paper are threefold:
\begin{itemize}
    \item First, we propose a VAE-based scalable and unsupervised machine-learning scheme that can determine the dynamic relay topology under the regional distribution of nodes and energy, thus overcoming the need for a labeled dataset or the lack of scalability of the model shown in the previous works using supervised learning.
    \item Second, we propose a novel backward-pass based rate evaluation method, called "Packet-Tracing"(PT), which can properly and concisely assess the output of our VAE scheme, thereby giving a proper direction for training without the need of a skilled and specified agent and environment modeling, which is inevitable in the previous works using DRL.
    \item Lastly, we propose an iterative balancing (IB) algorithm for time slot allocation over a time-division multiplexing access (TDMA) based system, which ultimately gives the solution on both topology planning and time slot planning for the formulated max-min optimization problem regarding fairness.
\end{itemize}

The remainder of this paper consists of the following.
In Section II, we define our system model and formulate a max-min optimization problem: our problem to solve. 
In Section III, from the problem formulation, we present a topology algorithm based on a VAE scheme and our novel backward-pass based rate evaluation method, and finally, a time slot allocation algorithm that maximizes the max-min fairness of nodes under a TDMA system.
In Section IV, we present operation details of our proposed scheme and perform numerical simulations and analysis which confirm the sub-optimality of our proposed scheme, leading to our conclusions in Section V.

\section{System Model and Problem Formulation}

\subsection{System Model}

\begin{figure}[t]
	\centering
	\includegraphics[width=0.72\linewidth]{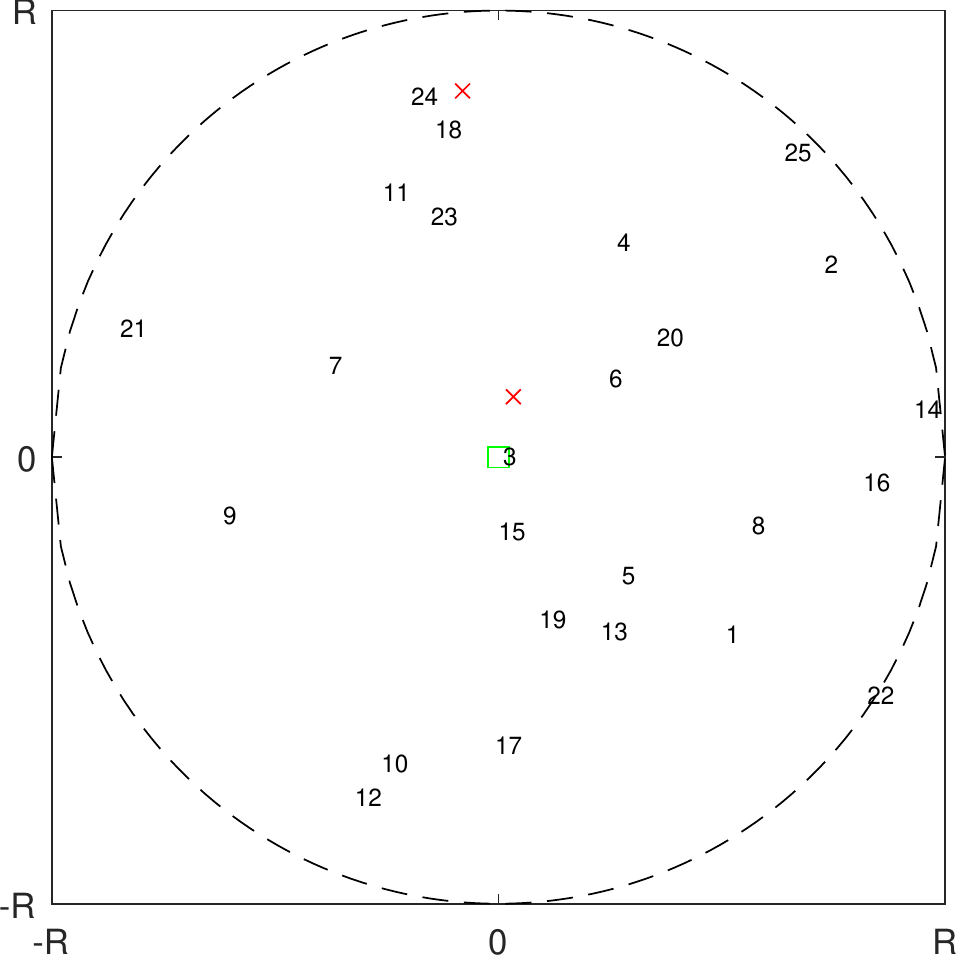} 
	\caption{Network example ({number} : IoT Nodes, \textcolor{red}{X} : Power Beacons, \textcolor{green}{$\square$} : Sink  )} 
	\label{NetExample}
 \vspace{-0.5 cm}
\end{figure}

In this subsection, we describe the system model.
Here, the uplink transmissions in an IoT sensor network during the time frame $T$ is considered.
In our system model consideration, $N_d$ IoT devices (shortly, nodes) are uniformly distributed in a circle of radius $R$ centred around the sink.
Let us denote $\mathbb{N}_d=\{1,2,...,N_d\}$ as the set associated with nodes.

By indexing the packet sink with $N_d+1$, we can define $\mathbb{N}=\mathbb{N}_d \cup \{N_d+1\}$ as the communicating node set.
In addition, a TDMA system is assumed in our system model.
TDMA is an effective system to alleviate the performance of nodes placed in an inferior power/position environment\cite{Sgora15}, by allocating more time slots within the uplink frame by taking spare time slots from the nodes placed in a superior power/position environment.
The considered time frame $T$ is divided into $N_d$ time slots and allocated to each IoT device, whose time slot for each node $n$ is indicated by $t_n \in (0, T)$.
During $t_n$, node $n$ has the opportunity to transmit.

Next, we consider energy harvesting (EH) for our system model. 
In our model, each node is self-powered by harvesting energy from PBs and stores the remaining energy in a rechargeable battery with limited capacity.
Power beacons(PB) of energy can be of any type (solar, RF, or others), and only one condition applies if the RF source is considered, whose frequency band should be different from the communication band used.
This condition alleviates the possibility that the energy harvesting signals cause additional interference to the sink.
In our problem, PBs are also randomly located in the cell with radius $R$, whose transmit power is defined as $P_b$.
Let us denote $\mathbb{N}_b=\{1,2,...,N_b\}$ as the set associated with PBs.
An example of the network is presented in Fig. \ref{NetExample}, consisting of $2$ PBs and $25$ nodes around the sink.
In this paper, we assume that the distribution of nodes and PBs follow a uniform distribution for simplicity. However, other practical distributions considering urban, indoor or transportation environments are also applicable.

We also consider RF as the power source type in our system model. In this case, the harvested energy of node $n$ which is sent from the PBs can be defined as
\begin{align}
    E_{\text{rec},n}=T\sum_{i \in \mathbb{N}_b}P_b |h_{i,n}|^2 d_{i,n}^{-\alpha},
\end{align}
where $d_{i,n}$ is the distance between node $i$ and node $n$, $h_{i,n}$ is the channel between node $i$ and node $n$ which follows a complex Gaussian distribution with zero-mean and unit variance, such that $h_{i,n} \sim \mathcal{CN}(0,1)$ and $\alpha$ is the path-loss exponent.
In our system model, we consider a discrete time block-fading model, where the channel state is constant for a time frame $T$ \cite{Atapattu16}, and also assume that the instantaneous channel state information for $h_{i,n}$ is available.

Then, the harvested energy at each node depends on which EH model is considered: either linear or nonlinear \cite{Clerckx19}.
One possible nonlinear EH model to consider is the sigmoidal model  \cite{Benkhelifa21}, which was shown to fit well with the experimental results.
In this paper, a linear model is considered for simplicity, but nonlinear models can also be applied.
When the linear model is considered, the harvested energy $E_n$ is given by $E_n = \eta E_{rec,n}$, where $\eta \in [0, 1]$ is the conversion efficiency.

We also assume that the harvested energy during $T$ is consumed for signal transmission in the $t_n$ time slot in the next $T$ time window, where $\mathbf{t}=\{t_n| \forall n\in \mathbb{N}_d\}$.
Then, the transmit power can be expressed as $P_n=\frac{E_n}{t_n}$.\footnote{For simplicity, we assume that $E_n$ is used only for the transmission, not for the transmission/reception electronic circuit \cite{Lim22}.
For an accurate model, such as the first order radio model, the circuit power consumption should considered, but it will only result in a slight difference to the degree of improvement in our proposal.}
In addition, in the TDMA system, the signal-to-noise ratio from a transmitting node $n$ to a receiving node $n'$, $\Gamma_{n,n'}$, can be expressed as  $\Gamma_{n,n'}=\frac{P_n |h_{n,n'}|^2 d_{n,n'}^{-\alpha}}{N}$ where $N$ is the noise power.

\begin{figure}[t]
	\centering
	\includegraphics[width=0.6\linewidth]{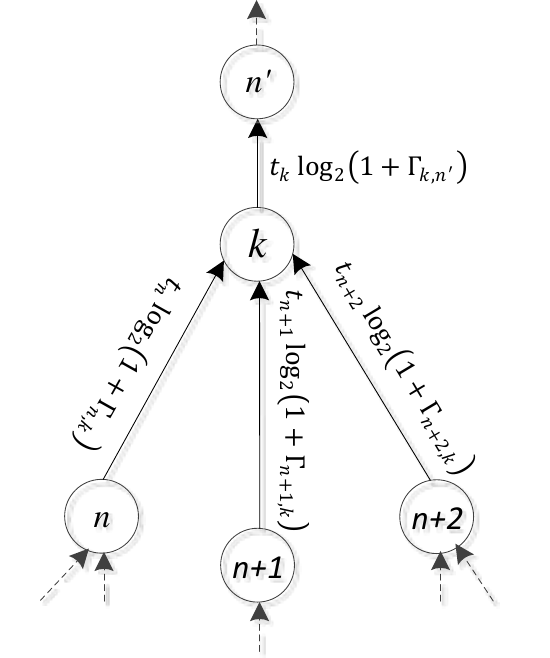} 
	\caption{Illustration of relaying data from lower nodes to a higher node} 
	\label{Concept}
 \vspace{-0.5 cm}
\end{figure}

Moreover, we define $\mathbf{c}=\{c_{n,n'}|\forall n \in \mathbb{N}_d, \forall n' \in \mathbb{N}\}$ whose element is $c_{n,n'} \in \{0, 1\}$ to represent the connection of the relay topology where $c_{n,n'}=1$ if node $n$ transmits its data and its child nodes' data to node $n'$; otherwise, $c_{n,n'}=0$.
Using the definition of $\mathbf{c}$ and Shannon's capacity, the amount of bits per frequency (bits/Hz) that node $k$ could transmit is defined as $\sum_{n'\in \mathbb{N}, n'\neq k} c_{k,n'} t_k \log _{2} (1+\Gamma_{k,n'})$, and the amount of bits/Hz that node $k$ receives is expressed as $\sum_{n\in \mathbb{N}_d, n\neq k} c_{n,k} t_n \log _{2} (1+\Gamma_{n,k})$ (as shown in Fig. \ref{Concept}).
Then, the difference between the two could be considered as the amount of bits/Hz the node $k$ itself can transmit, denoted as $B_k(\mathbf{c},\mathbf{t})$, and can be expressed as follows

\begin{align}
B_k(\mathbf{c},\mathbf{t}) = 
\sum_{n'\in \mathbb{N}} c_{k,n'} t_k \log _{2} (1+\Gamma_{k,n'}) \nonumber\\
- \sum_{n\in \mathbb{N}_d, n\neq k} c_{n,k} t_n \log _{2} (1+\Gamma_{n,k}) .\label{B_calculate}
\end{align}

\subsection{Problem Formulation}
In this subsection, we formulate the optimization problem based on the system model described above.
The optimization problem for maximizing the minimum of $B_k$ among all nodes, $\mathbb{N}_d$, can be formulated as follows:
\begin{align}
   \text{(P1)}  \quad \max_{\mathbf{c}, \mathbf{t}} \, \min_{k} \quad & B_{k}(\mathbf{c},\mathbf{t})\label{p1}\\
\textrm{s.t.} \quad & \frac{1}{T} \sum_{n \in \mathbb{N}_d} t_n =1, \label{c1}\\
  & \sum_{n'\in \mathbb{N}} C_{n,n'} =1,\,\, \forall n \in \mathbb{N}_d,\label{c3}\\
  & (\mathbf{C}^{N_d})_{n,N_d+1} = 1, \forall n \in \mathbb{N}_d \label{c4} .
\end{align}
where $\mathbf{C}$ is the extended adjacency matrix of topology, whose $C_{n,n'} \in \{ 0,1 \}$, $\forall n\in \mathbb{N}_d$, $\forall n'\in \mathbb{N}$, $C_{N_d+1,n}=0$, $\forall n\in \mathbb{N}_d$ and  $C_{N_d+1,N_d+1}=1$.
(\ref{c1}) guarantees that the sum of $t_n$ for all nodes is $T$.
(\ref{c3}) indicates that the outward link of each node is only connected to one of the other nodes including the sink.
(\ref{c4}) means that every node is connected to the sink in the end.
\footnote{According to the adjacency matrix $\mathbf{A}$, $n$-th matrix multiplication of adjacency the matrix element $(\mathbf{A}^n)_{i,j}$ indicates whether $n$-hop paths exists from node $i$ to node $j$.
Since the adjacency matrix mentioned here has been extended to have a self-cycle at the packet sink,we include all $m$-hop paths while regarding the extended adjacency matrix multiplication $(\mathbf{A}^n)_{i,j}$ where ${m} \leq{n}$.}

This optimization problem is NP-hard; therefore, an exact algorithm demands enormous computational effort.
However, by decoupling the problem into two sub-problems, we can transform the original problem with $\mathbf{c}$ and $\mathbf{t}$ to the single variable problem with $\mathbf{c}$.

For a given $\mathbf{c}$, the problem (P1) can be transformed to
\begin{align}
     \text{(P1-1)}\quad \max_{\mathbf{t}} \, \min_{k} \quad & B_{k}(\mathbf{c},\mathbf{t})\label{p2}\\
\textrm{s.t.} \quad & \frac{1}{T} \sum_{n \in \mathbb{N}_d} t_n =1. 
\end{align}
Since the result of the IB time slot allocation algorithm proposed in the next section satisfies $\max_{k} B_{k}(\mathbf{c},\mathbf{t})=\min_{k} B_{k}(\mathbf{c},\mathbf{t})$ under the given $\mathbf{c}$, we can always find $\mathbf{t}$ satisfying (P1-1) by using this algorithm.
In other words, $\mathbf{t}$ could be considered as a variable determined by $\mathbf{c}$.
If we define the optimal value of (P1-1) as $B_{IB}(\mathbf{c})$, the original problem (P1) could be expressed as below:
\begin{align}
    \text{(P2)}\quad \max_{\mathbf{c}} \, \quad & B_{\text{IB}}(\mathbf{c})\label{p3}\\
\textrm{s.t.} \quad   & \sum_{n'\in \mathbb{N}} C_{n,n'} =1,\,\, \forall n \in \mathbb{N}_d,\label{C11}\\
  & (\mathbf{C}^{N_d})_{n,N_d+1} = 1, \forall n \in \mathbb{N}_d. \label{C22}
\end{align}
Then, problem (P2) is equivalent to problem (P1).
We can also ignore the constraints (\ref{C11}) and (\ref{C22}) since $\mathbf{c}$ does not satisfy those constraints, and gives $B_{\text{IB}}(\mathbf{c}) \leq 0$. 
Unfortunately, for solving (\ref{p3}), no computationally efficient method could be applied without involving exhaustive search which has a time complexity of at least $\mathcal{O}((N_d+1)^{N_d-1})$.
To avoid such time complexity limitation, we propose a VAE based Machine-Learning algorithm in the next section.


\section{Description of the algorithm}
In this section, we propose a VAE scheme to find $\mathbf{c^*}$ and use our novel PT algorithm to assess it. We also propose the IB time slot allocation to find $\mathbf{t^*}$, thereby obtaining the sub-optimal solution over both $\mathbf{c}$ and $\mathbf{t}$.

\subsection{Time Slot Allocation Algorithm}

\begin{algorithm}
\caption{Iterative Balancing (IB) algorithm for $\mathbf{t}^*$}\label{alg_TSalloc}
\begin{algorithmic}[1]
\STATE $t_{n}=\frac{T}{N_d}, \, \forall n\in \mathbb{N}_d$
\STATE $\delta_1 =\infty$
\WHILE {$\epsilon_1 < \delta_1$}
\STATE $i^*=\underset{k}{\arg \, \max}\, B_k(\mathbf{c},\mathbf{t})$, $j^*=\underset{k}{\arg \, \min}\, B_k(\mathbf{c},\mathbf{t})$
\STATE $\Delta=t_{i^*}$
\STATE $\delta_2=\delta_1$
\WHILE {$\Delta > 2\epsilon_2$ \& $\epsilon_1 < |\delta_2|$} 
\STATE $\Delta =\Delta/2$
\IF{$\delta_2 > 0$}
\STATE $t_{i^*}=t_{i^*}-\Delta$
\STATE $t_{j^*}=t_{j^*}+\Delta$
\STATE $\delta_2 =B_{i^*}(\mathbf{c},\mathbf{t})- B_{j^*}(\mathbf{c},\mathbf{t})$
\ELSE
\STATE $t_{i^*}=t_{i^*}+\Delta$
\STATE $t_{j^*}=t_{j^*}-\Delta$
\STATE $\delta_2 =B_{i^*}(\mathbf{c},\mathbf{t})- B_{j^*}(\mathbf{c},\mathbf{t})$
\ENDIF
\ENDWHILE
\STATE $B_{\text{max}} =\underset{k}{\max}\, B_k(\mathbf{c},\mathbf{t}), B_{\text{min}}=\underset{k}{\min}\, B_k(\mathbf{c},\mathbf{t})$
\STATE $\delta_1 =B_{\text{max}}-B_{\text{min}}$
\ENDWHILE
\STATE $B_{\text{IB}}(\mathbf{c})=B_{\text{min}}$
\end{algorithmic}
\end{algorithm}

In this subsection, we propose an IB time slot allocation algorithm which derives the optimal $\mathbf{t}$ satisfying (P1-1) under a given $\mathbf{c}$, thus maximizing $B_{\text{min}}$.
Firstly, we define the hyper-parameter for algorithm, $\epsilon_1$ and  $\epsilon_2$, which are the upper bound of the difference between $B_{\text{max}}$ and $B_{\text{min}}$ and the minimum allocatable time slot, respectively.
Additionally, the condition $\epsilon_2 \ll T$ and $ B_k(\mathbf{c},2 \epsilon_2) < \epsilon_1$ for $\forall k \in \mathbb{N}_d$ 
must be met to guarantee the convergence of our algorithm.

The key principle of the proposed algorithm is to repeatedly reduce the difference between $B_{\text{max}}$ and $B_{\text{min}}$.
First, under the initial $\mathbf{t}$ assignment, we find $i^*=\underset{k}{\arg \, \max}\, B_k(\mathbf{t})$ and $j^*=\underset{k}{\arg \, \min}\, B_k(\mathbf{t})$.
Then, the bisection algorithm is repeatedly applied to make the amount of bits/Hz between node $i^*$ and $j^*$, $\delta_2$, less than $\epsilon_1$, or to make the time slot to be delivered between node $i^*$ and $j^*$, $\Delta/2$, less than $\epsilon_2$. 
The bisection algorithm then sets $\Delta = \Delta/2$ and delivers the time slot between node $i^*$ and $j^*$ by reducing the time slot of the node which has a superior $B_k(\mathbf{t})$ by $\Delta$ and instead allocates that reduced $\Delta$ to the node which has an inferior $B_k(\mathbf{t})$.

When the bisection algorithm of the two selected nodes is finished, the same procedure is repeated by obtaining the newly founded node $i^*$ and node $j^*$ under the condition depicted above.
This procedure is performed until $\delta_1=B_{\text{max}}-B_{\text{min}}$ becomes less than than $\epsilon_1$.
A detailed description of the process described is given in Alg. \ref{alg_TSalloc}.
Furthermore, the convergence proof of Alg. \ref{alg_TSalloc} is given as follows.

\begin{figure*}
	\centering
	\includegraphics[width=1.0\linewidth]{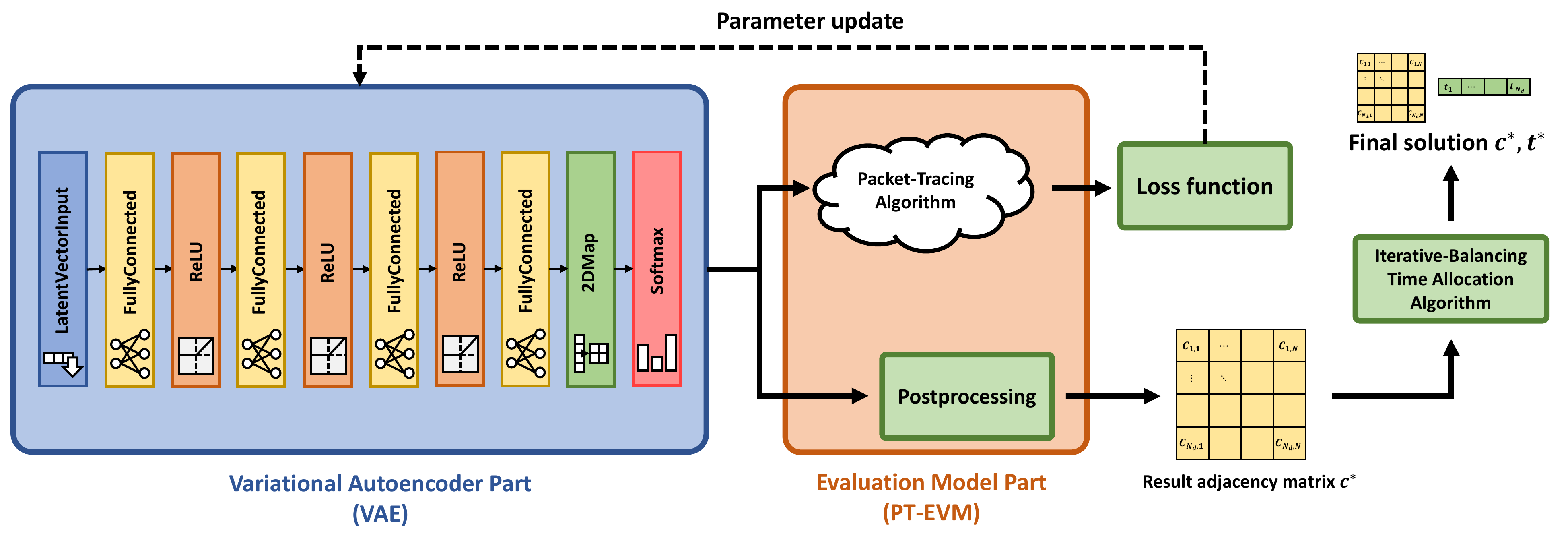} 
	\caption{Structure flow chart of the proposed scheme to find $\mathbf{c^*}$ and $\mathbf{t^*}$.} 
	\label{struct}
 \vspace{-0.5 cm}
 \end{figure*}

\begin{pro}
Let us define $B_{\text{max},k}$, $B_{\text{min},k}$ and $\delta_{1,k}$ as $B_{\text{max}}$, $B_{\text{min}}$ and $\delta_1$ in the $k$-th iteration of the outer loop which consists of lines 3 to 21 in Alg. \ref{alg_TSalloc}, respectively.
Then, $\delta_{1,k} = (B_{\text{max},k}-B_{\text{min},k}) \leq \epsilon_1$ is satisfied for $\exists k$ in the total outer loop iteration count domain.
\end{pro}

\begin{proof}
Let us define the time slot allocated to the node $i$ and the amount of bits/Hz of node $i$ at the beginning of the $k$-th iteration as $t_{i,k}$ and $B_{i,k} = B_{i}(t_{i,k})$, respectively.
In this notation, $B_{\text{max},k}$ and $B_{\text{min},k}$ at the beginning of the $k$-th iteration can be written with $B_{i^*}(t_{i^*,k})$ and $B_{j^*}(t_{j^*,k})$, respectively.

Conversely, termination of the inner loop consists of lines 7 to 18 in Alg. \ref{alg_TSalloc} indicating that the condition $\Delta \leq 2\epsilon_2$ or the condition $\epsilon_1 \geq |\delta_2|$ has been met.
If at least one iteration of the inner loop has been performed, we can find such proper positive constant $\sigma$ among all outer loop iterations that satisfy
$\sigma < \text{min}(B_{\text{max},k} - B_{i^*,k+1} ,B_{j^*,k} - B_{\text{min},k+1})$ and $\sigma < \epsilon_1$ according to lemma \ref{lemma_1}.
Since $\Delta > \epsilon_2$ is always guaranteed, it eventually guarantees the existence of constant $\sigma$ and $N_\delta < N/2$ satisfying $\delta_{1,k}\geq \delta_{1,k+N_{\delta}} + \sigma$ for $\forall k$ in the total outer loop iteration count domain.

Thus, in this case, the two expressions $\delta_{1,k} \geq \delta_{1,k+N_{\delta}} + \sigma$ and $\delta_{1,k} \geq 0$ are always satisfied for $\forall k$ and constant $\exists \sigma > 0$ and $\exists N_{\delta} < N/2$ eventually leads to the satisfaction of the condition $\delta_{1,k} \leq \epsilon_1$ for some $k$.

However, such termination condition of the inner loop could be met at the beginning of the inner loop iteration; hence, the iteration of the inner loop may not be performed.
Nevertheless, neither condition can be met at the beginning of the first iteration of the inner loop.

Firstly, satisfying the former condition $\Delta \leq 2\epsilon_2$ means that $t_{i^*,k} \leq 2\epsilon_2$ since we set the initial $\Delta = t_{i^*}$.
Subsequently, the inequality $\delta_{1,k} \leq B_{i^*}(t_{i^*,k}) \leq B_{i^*}(2 \epsilon_2) < \epsilon_1 $ is justified since condition $B_i(\mathbf{c},2 \epsilon_2) < \epsilon_1, \forall i \in \mathbb{N}_d$ is given.
Next, satisfying the latter condition $\epsilon_1 \geq |\delta_2|$ implies that $\epsilon_1 \geq \delta_{1,k}$ has been met, since we set the initial $\delta_2 = \delta_{1,k}$.

Therefore, it can be shown that both cases already satisfy the termination condition of the outer loop, eventually satisfying the proposition.

\end{proof}

\begin{lem} \label{lemma_1}
A positive constant $\sigma$ satisfying $\sigma < B_{i}(t_A)-B_{i}(t_B)$ and $\sigma < \epsilon_1$ for $\forall i \in \mathbb{N}_d$ always exists while $B_{n}(t_n) = t_n \text{log}_2(1+\Gamma_{n,n'})$ is given with $\Gamma_{n,n'}=\frac{P_n |h_{n,n'}|^2 d_{n,n'}^{-\alpha}}{N}$ and $P_n = \frac{E_n}{t_n}$, respectively. This satisfies $\epsilon_2 \leq t_A-t_B $ and $c_{n,n'} = 1$.

\end{lem}

\begin{proof}
Let us denote $B_i(t) = t \text{log}_2(1+\frac{\tau_i}{t})$ for some constant $\tau_i$ since $E_n, |h_{n,n'}|^2, d_{n,n'}^{-\alpha}, N$ is constant. Next, we denote the node index $i^*$ which satisfies the condition $B_{i^*}(T+\epsilon_2) - B_{i^*}(T) \leq B_{i}(T+\epsilon_2) - B_{i}(T)$ for $ \forall i \in \mathbb{N}_d$. 

Note that function $B_i(t)$ is a monotonic increasing function and $\frac{\partial^2}{\partial t^2} B_i(t) < 0$ in domain $t \in \mathbb{R^+_*}$.
Then, the inequality $B_i(T+\epsilon_2) - B_i(T) < B_i(t_A) - B_i(t_B), \forall i \in \mathbb{N}_d$ is justified since condition (\ref{c1}) exists, restricting $t$ in $0<t<T$.

Also, since condition $B_i(\mathbf{c},2 \epsilon_2) < \epsilon_1, \forall i \in \mathbb{N}_d$ is given, the inequality $B_{i^*}(T+\epsilon_2) - B_{i^*}(T) < B_{i^*}(\mathbf{c},2 \epsilon_2) - B_{i^*}(\mathbf{c},\epsilon_2) < B_{i^*}(\mathbf{c},2 \epsilon_2) < \epsilon_1$ is justified.

Thus, it is sufficient to set the positive constant $\sigma$ satisfying the lemma above as $B_{i^*}(T+\epsilon_2)-B_{i^*}(T)$.

\end{proof}

\subsection{Machine Learning Based Topology Algorithm}

In this subsection, we propose a machine-learning based VAE model and a backward-pass based rate evaluation model called the Packet-Tracing evaluation model (PT-EVM), inspired by the ray tracing method used in 3-D graphics rendering.
Our proposed model to find $\mathbf{c^*}$ consists of two parts: firstly, deriving $\mathbf{c}^*$ using the VAE part and assessing the derived $\mathbf{c}^*$ with the PT-EVM part, and secondly, training the VAE part with the assessment from the PT-EVM part.

Since (P2) has non-linearity and discontinuity characteristics due to the nature of the Shannon-capacity formula and the mixed-integer problem, it is difficult to find the sub-optimal $\mathbf{c}^*$ using analytical methods.
Given these considerations, we propose machine-learning as a method to overcome these limitations.
However, there are some restrictions in applying supervised-learning in our formulated problem. The dataset used for supervised learning, such as a pair set of positional distributions among nodes and PBs and optimal topology, cannot be obtained due to the characteristics of our formulated problem.
Moreover, it is difficult to build such dataset in an exhaustive way due to the exponential growth of computational time with the number of nodes.

One alternative to consider is a reinforcement learning method that retrieves rewards from interacting with a physical network by applying resulting topology from the algorithm to the physical network, then analyzes real achievable rates as a reward.
However, this method has some drawbacks.
When a model is in the learning phase, topology output is imperfect, which is an essential procedure for receiving a reward to use in reinforcement learning. As a result, initial performance degradation while learning in practical use is inevitable. 
Moreover, as depicted in Section I, the IoT network environment has ultra-low speed intermittent transmission properties, meaning that there is a very low frequency of rewards which can be used in reinforcement learning, thus maximizing the drawback depicted above. These drawback make applying reinforcement learning improper.

Taking these limitations into consideration, we propose a VAE based generative model which can be used in an unsupervised learning manner.
Typically, generative models are used to resemble an observed dataset and embed the dataset to a high-dimensional space, called the latent vector space, thus gaining the ability to generate new data points based on observed distribution by simply modifying the newly latent vector.
Considering this, our generative model and unsupervised-learning hybrid scheme slightly modifies these original structures, which explores latent vector space and its mapped data point space to derive the sub-optimal solution $\mathbf{c}^*$.
However, to realize the proposed concept depicted above, a proper assessment of the resulting output $\mathbf{c}$ is necessary, which enables exploration of the latent vector space and derivation of the sub-optimal solution.

Therefore, a simulation-based assessment for the resulting topology is essential, even though it requires both brevity but sufficient detail to provide precision network modeling. 
However, conventional network modeling and simulation (M\&S) tools such as NS-3 or the Riverbed Modeler requires excessive computing overhead since they simulate the entire network operation, such as buffer management, TCP/IP protocol procedure, wireless channel emulation, packet switching among interacting interfaces, etc.
Such simulation overhead is unnecessary since our assessment requirement is sufficient to examine the achievable rate over nodes with regards to the congestion and distribution of PBs and nodes.
One imaginable solution is to subtract the sum of the inbound capacity from the sum of the outbound capacity on the basis of the definition of Shannon's capacity, depicted as (\ref{B_calculate}). 
Unfortunately, this naive solution cannot consider how much of the rate(packet) will be inbound to a specific node and how much of the rate(packet) can be handled in the destination node regarding congestion.

Since these limitations and requirements cannot be fully reflected through typical solutions, we propose a PT algorithm to enable congestion-aware rate assessment. 
In our PT algorithm, we sequentially assign the achievable rate of the node starting with the packet sink and towards the backward-pass manner, like the ray-tracing algorithm in 3-D graphics rendering.
Ray-tracing works by tracing a path from a viewpoint (camera) to an object in virtual 3-D graphics space, and the ray propagates under photonics simulation, eventually hitting the light source. 
Finally, the ray shows the valid path for the light to propagate \cite{Steven05}.
Since the ray-tracing algorithm does not consider rays which don't reach the camera, the computation cost is drastically reduced.
Similarly, the backward-pass property of our packet-tracing algorithm provides both feasible and concise rate assessment similar to that of the ray-tracing algorithm.

As shown in Fig. \ref{struct}, the proposed unsupervised-learning model consists of two parts: mainly the VAE part and the PT-EVM part. 
To summarize, this proposed concept can be implemented by entering a random latent vector input into VAE, assessing the output $\mathbf{c}$, and updating the parameters of VAE using these assessments iteratively.
This procedure eventually focuses and maps prior latent vector space to the sub-optimal solution data point and draws $\mathbf{c}^*$ without the dataset necessary for supervised learning.
A detailed description of each part of scheme is depicted in the subsections below.

\subsubsection{VAE structure}
The VAE part generates a topology for a given latent vector input.
The left side of Fig. \ref{struct} depicts the VAE structure as 4 fully connected (FC) layers and 3 rectified linear unit (ReLU) layers between each FC layer, followed by a 2Dmap layer and softmax layer.

Firstly, in the FC layer, each input element $x_{i}$ is processed into the output element $y_{j}$ by the following equation:

\begin{align}
y_{j} = \sum_{i\in \mathbb{N}_I} w_{i,j}x_{i}+b_{j}, \forall{j} \in \mathbb{N}_O\label{y_j}
\end{align}

where $i$ and $j$ are the input and output element index, respectively. $\mathbb{N}_I$, $\mathbb{N}_O$, $\mathbf{w}$, and $\mathbf{b}$ are the number set of the input element, output element, weights, and biases of the layers, respectively. 
For each FC layer, the number of output elements is set in an arithmetic sequence manner, whose number of output elements of the last layer is fixed to $N_d(N_d+1)$ for topology mapping (i.e., each node in $\mathbb{N}_d$ can select the upper link among $\mathbb{N}$).

Between the FC layer, an ReLU layer is used as an activation function to provide handling on the non-linear property of the formulated problem.
In the ReLU layer, each input element $x$ is processed into the output element $y$ by the following equation: 

\begin{align}
y = \max(x,0) \label{relu}.
\end{align}

Next, a 2DMap layer is considered as a means of mapping the FC layer output to adjacency matrix form. 
The 2DMap layer maps the 1-Dimension input vector $x$ totaling a $N_d(N_d+1)$ element count to $N_d$ by the $N_d+1$ adjacency matrix, denoted as  $y_{i,j}$, as follows:

\begin{align}
y_{i,j} = x_{(N_d+1)(i-1)+j}, \forall{i} \in \mathbb{N}_d , \forall{j} \in \mathbb{N} \label{y_ij}.
\end{align}

Finally, a softmax layer is applied to make the sum of each sequential row equal to 1, so that the total outward connectivity summation for each node is equal to 1.
Furthermore, the softmax layer functions as an activation function, thus enhancing the ability to handle the non-linear property. 
The softmax layer operates for the input element $x_i$ and the output element $y_i$ in each sequential row, as follows: 

\begin{align}
y_i = \frac{\exp{x_i}}{\sum_{j\in \mathbb{N}} \exp{x_j}} , \forall i \in \mathbb{N}\label{y_i}.
\end{align}

Therefore, the VAE part results in an adjacency matrix $\mathbf{c}$ as the final output, which determines the entire topology connectivity configuration as described above.

\subsubsection{PT-EVM structure}
The PT-EVM part evaluates the resulting topology from the VAE part. 
The right side of Fig. \ref{struct} depicts the PT-EVM structure consisting of the PT algorithm module to achieve the rate required for training loss function calculation, and the post-processing module to post-process the inference result.

Similar to the ray-tracing algorithm, our PT algorithm back-propagates and assess the achievable rate budget in a backward process, starting from the packet sink.
A detailed description of the procedure of our PT algorithm is as follows.
Firstly, in line 9 of Alg. \ref{RatePT}, the focused node (i.e., a packet sink in the initial step) starts the algorithm with the granted rate budget (i.e., $\infty$ in the initial step) given in the previous step. 
Next, the focused node takes its share of rate budget, $R_{\text{self}}$, which depends on the total inbound connection status and rate budget (line 12).
We turned off auto-differentiation calculation when calculating $I, B, R_{\text{self}}$, thus preventing the assessment of the other nodes to be affected.

Subsequently, in line 17, the focused node gives its remaining rate budget for inbound nodes.
Basically, each inbound node budget is set to their link rate.
However, if the remaining budget of the focused node, $B-R_{\text{self}}$, is less than $I$, this indicates that congestion is happening on the focused node.
In this case, each inbound node budget is reduced in proportion to $\frac{B-R_{\text{self}}}{I}$.

Then, the next unallocated node is called respectively in line 20 and these steps are repeated until a termination condition (e.g., remaining allocatable rate budget drops below a certain threshold, $B_{\text{th}}$) is satisfied.
This recursive algorithm spreads out through the inbound connections, and the sum of all granted rate budgets over various packet-rays pass a specific node, it implies the achievable rate of the node.
With this manner of calculation, we can guarantee that all granted achievable rates can be reached into the packet sink since we calculated in backward-pass from the packet sink, like in the ray of ray-tracing.

After finalizing the backward calculation of all packet-rays, we calculate the net packet-ray for all edges of the topology since a single packet-ray does not consider other packet rays; thus, net rate calculation is essential to PT-EVM modeling (lines 5-7).
Alg. \ref{RatePT} depicts the implementation details of the proposed packet-tracing algorithm, where $I$, $B$, and $R_{\text{self}}$ are the sum of the inbound rate, the granted rate budget through the packet-ray, and the allocated rate of the called node, respectively. 
$t$ is the time slot scale factor for the simulation algorithm and $B_{\text{th}}$ is the rate budget threshold for determining termination condition.

\begin{algorithm}
\caption{Calculation of $\mathbf{R}_{\text{sim}}$}\label{RatePT}
\begin{algorithmic}[1]
\STATE \textbf{main}
\STATE{$\mathbf{f} =\{f_{i}|f_{i}=0, \forall i\in \mathbb{N}_d\}$}
\STATE{$\mathbf{R} =\{R_{i,j}|R_{i,j}=0, \forall i\in \mathbb{N}_d, \forall j\in \mathbb{N}\}$}
\STATE $\mathbf{R} = \text{ RatePT}(N_d+1,\infty,\mathbf{f})$
\STATE $(\mathbf{R}_{\text{net}})_{i,j} = $relu$(R_{i,j}-R_{j,i}) , \forall i, j\in \mathbb{N}_d$
\STATE $(\mathbf{R}_{\text{net}})_{i,N_d+1} = R_{i,N_d+1} , \forall i \in \mathbb{N}_d$
\STATE $(\mathbf{R}_{\text{sim}})_{i} = \sum_{j}{(\mathbf{R}_{\text{net}})_{i,j}} , \forall i\in \mathbb{N}_d$
\STATE \textbf{end main}
\\\hrulefill
\STATE \textbf{function} RatePT$(n, B, \mathbf{f})$
\STATE{$\mathbf{R} =\{R_{i,j}|R_{i,j}=0, \forall i\in \mathbb{N}_d, \forall j\in \mathbb{N}\}$}
\STATE \textbf{if} $B<B_{\text{th}}$, \textbf{then} terminate
\STATE $R_{\text{self}} = \frac{B}{1+\sum_{j\in \mathbb{N}_d} c_{j,n}}$
\STATE $P_j=\frac{E_j}{t},  \Gamma_{j,n}=\frac{P_j d_{j,n}^{-\alpha}}{N} , \forall j \in \mathbb{N}_d$
\STATE $L_{j} = c_{j,n} t \log _{2} (1+\Gamma_{j,n}) , \forall j \in \mathbb{N}_d$
\STATE $I = \sum_{j\in \mathbb{N}_d} L_{j}$
\STATE {Turn off auto-differentiation trace on $I, B, R_{\text{self}}$}
\STATE $R_{j,n} = L_{j}  \min(1,\frac{B-R_{\text{self}}}{I}) , \forall{j}\in \mathbb{N}_d$
\STATE $f_{n} = 1$\\
\STATE{\textbf{for} $j \gets \mathbb{N}-\{ n \}$\,\,\textbf{do}}
\STATE{\textbf{\,\,\,\,if} $f_{j}=0$, \textbf{then} $\mathbf{R} = \mathbf{R} +  \text{RatePT}(j,R_{j,n},\mathbf{f}), \forall{i} \in \mathbb{N}_d$}
\RETURN {$\mathbf{R}$}
\STATE \textbf{end function}
\end{algorithmic}
\end{algorithm}

\subsubsection{Training and Inferencing of the Proposed VAE Part}

After assessing the achievable rate of each node through the PT-EVM part, we calculate the training loss function using the calculated achievable rate, and use its value in our unsupervised learning for the VAE part. 
Since the PT-EVM part performs auto-differentiation tracing along the calculation, we can use the output of the PT-EVM part into training the loss function calculation for the VAE part and use it directly as a conventional neural-network training method (e.g., stochastic gradient descent).
Our training of model is specific to the given topology which enables compatibility over an arbitrary node and PB configuration (e.g., the number of nodes/PBs and the positional distribution of nodes/PBs).
The unsupervised learning proceeds by minimizing the training loss function value $L$ which is depicted as,
\begin{align}
{L} = \frac{1}{N_d}\sum_{i \in \mathbb{N}_d}{\text{exp}(-R_{\text{sim}})_i} \label{Lossfunction}
\end{align}
where $(R_{\text{sim}})_i$ is the rate of node $i$ calculated through the PT-EVM part.

Note that the training loss function we propose does not indicate the degree of similarity or sagging with the optimal answer, but is only a performance indicator evaluated by the PT structure.
That is, the training loss value cannot be zero, as such situation implies that $(R_\text{sim})_i = \infty, \forall i \in \mathbb{N}_d$.
Also, optimal training loss may vary for various IoT network settings (e.g. the optimal training loss $L_\text{opt}$ may have a value of about 0.8 while the optimal solution of the target IoT network configuration has $(R_\text{sim})_i=0.2, \forall i \in \mathbb{N}_d$).

Additionally, since algorithms learn through interactions that occur in a specific target IoT network configuration, there is no separate training or validation dataset in this methodology: there are no applicable validation loss metrics for this methodology. 
Alternatively, to validate the training epoch without validation loss, the actual performance indicator($B_\text{min}$) was tracked during the training phase in our analysis described below.

Next, the parameters of the VAE part are updated towards minimizing the training loss function (\ref{Lossfunction}) using the Adaptive Moment Estimation (ADAM) optimizing algorithm \cite{Kingma14} as follows:
\begin{align}
g_{t+1} &\gets \nabla_{\theta}L(\theta_{t}), \label{ADAMstart}\\
m_{t+1} &\gets \beta_1 m_{t} + (1-\beta_1) g_{t+1},\\
v_{t+1} &\gets \beta_2 v_t + (1-\beta_2) g^2_{t+1},\\
\hat{m}_{t+1} &\gets m_{t+1} / (1-\beta^{t+1}_1),\\
\hat{v}_{t+1} &\gets v_{t+1} / (1-\beta^{t+1}_2),\\
\theta_{t+1} &\gets \theta_{t} -\kappa\hat{m}_{t+1}/(\sqrt{\hat{v}_{t+1}}+\epsilon),\label{ADAMend}
\end{align}
where $\theta$, $t$, $m$, $v$, $\hat{m}$, and $\hat{v}$ are the learnable parameters (such as weights and biases in the neural network), timestep, biased first moment, biased second raw moment, bias-corrected first moment, and bias-corrected second raw moment, respectively. 
$\kappa$, $\beta_1$, $\beta_2$, and $\epsilon$ are the hyper-parameters for the model, which means the learning rate, decay rate for the first moment, decay rate for the second raw moment, and the small scalar value used to prevent the divide by zero error, respectively.

\cite{Doersch16} has shown that an arbitrary simple prior latent vector distribution can be converted to a specific target manifold distribution since the first few layers of the VAE structure adequately provide such conversion. Thus, we use prior latent vector distribution as the uniform distribution, $\mathcal{U}(0,1)$, for each element during training, and measure the training loss function for the sampled result. We then finally update the parameters of the VAE structure using the loss function measured.
Through these iterations, the VAE structure maps the simple prior latent vector distribution to the specific point of the latent space which embeds the best topology that VAE can provide. It then maps that specific latent space point to the best topology.

After training, we obtain our final inference result for a specific topology by post-processing the results as follows. 
The post-processing step finds the max index ($j^*$) in an adjacency matrix over a row, which is defined by $C_{i,j^*} \geq C_{i,k},\,\forall k \in \mathbb{N}-\{j^*\}$.
Then, it zeros out other outbound connections that are inferior to one which has the maximum connectivity, satisfying the condition $C_{i,j} \in \{ 0,1 \}, i \in \mathbb{N}_d , j \in \mathbb{N} $.
This post-processing step forces the topology result to have only one outbound connection.
Such computation can be described for $ i \in \mathbb{N}_d$ as below:
\begin{align}
C_{i,j^*} &\gets 1 , \\
C_{i,k} &\gets 0 ,\, \forall k \in \mathbb{N}-\{j^*\}.
\end{align}


\begin{figure}
	\centering
 \mbox{
 \subfigure[Epoch=0]{\label{Work1-1}
	\includegraphics[width=0.3\linewidth]{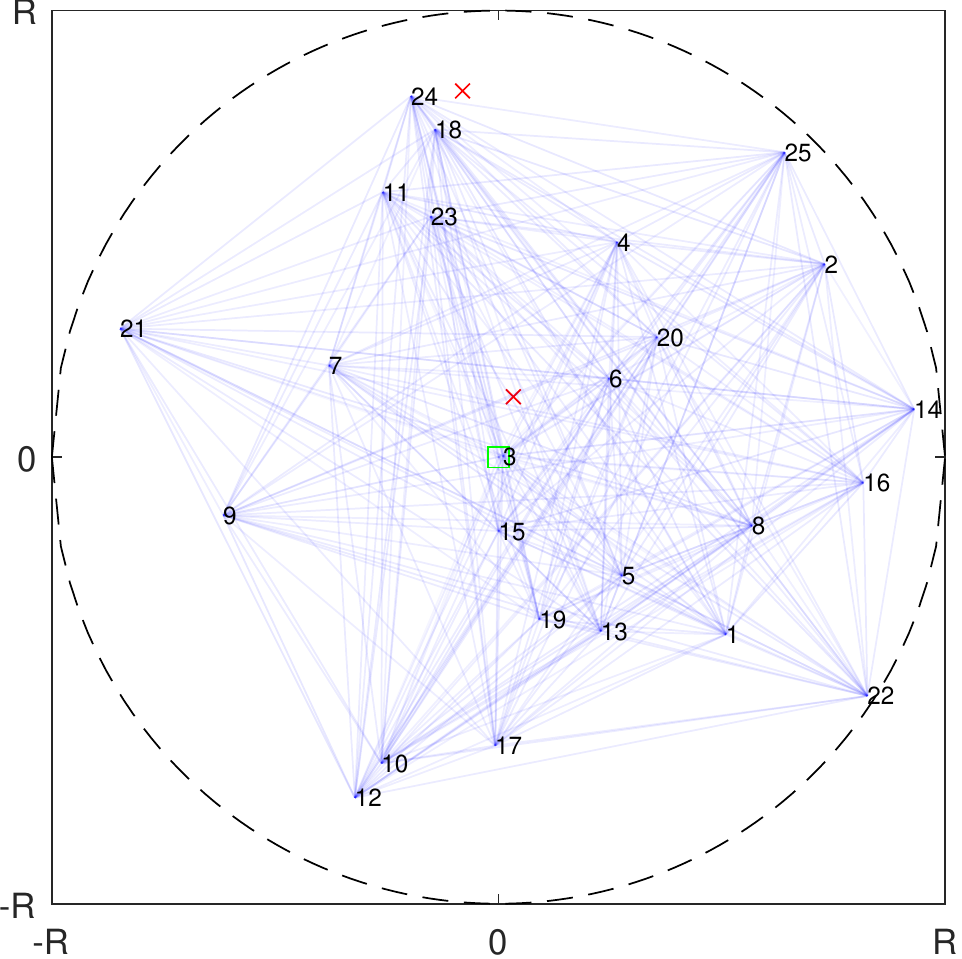} 
 }
  \subfigure[Epoch=10]{\label{Work1-2}
 \includegraphics[width=0.3\linewidth]{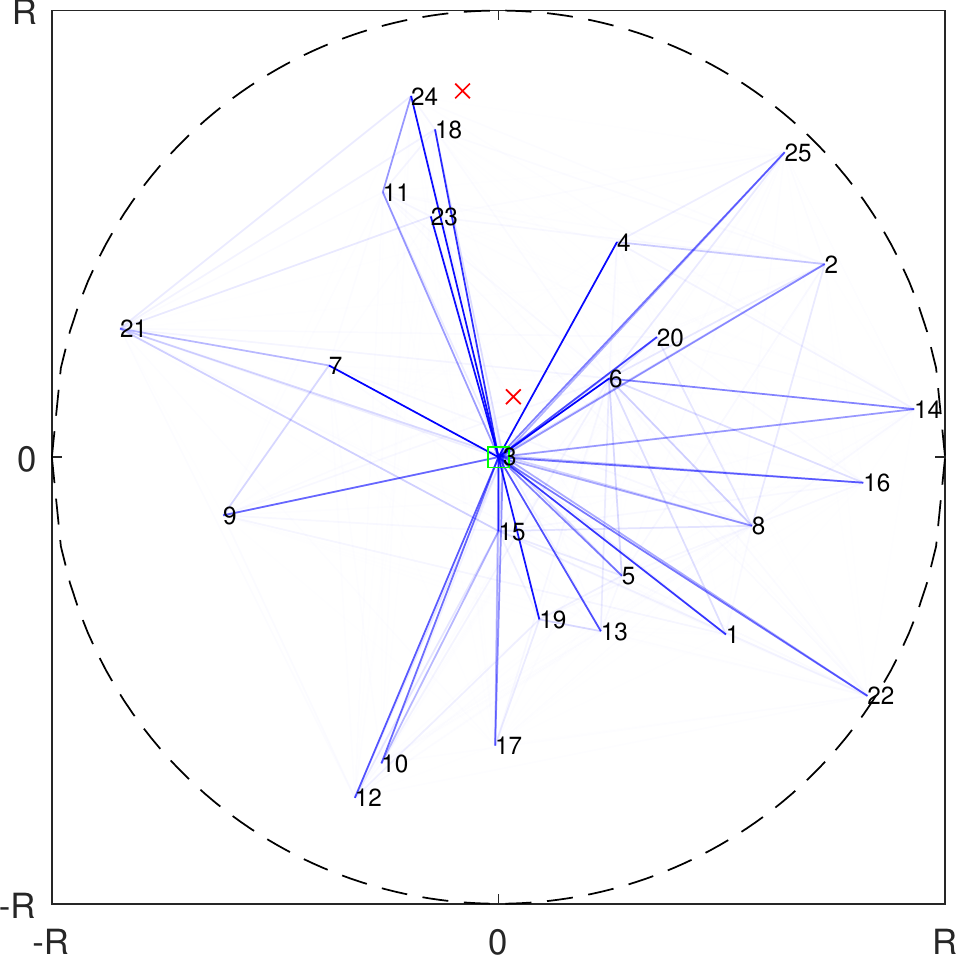} 
 }
  \subfigure[Epoch=50]{\label{Work1-3}
 \includegraphics[width=0.3\linewidth]{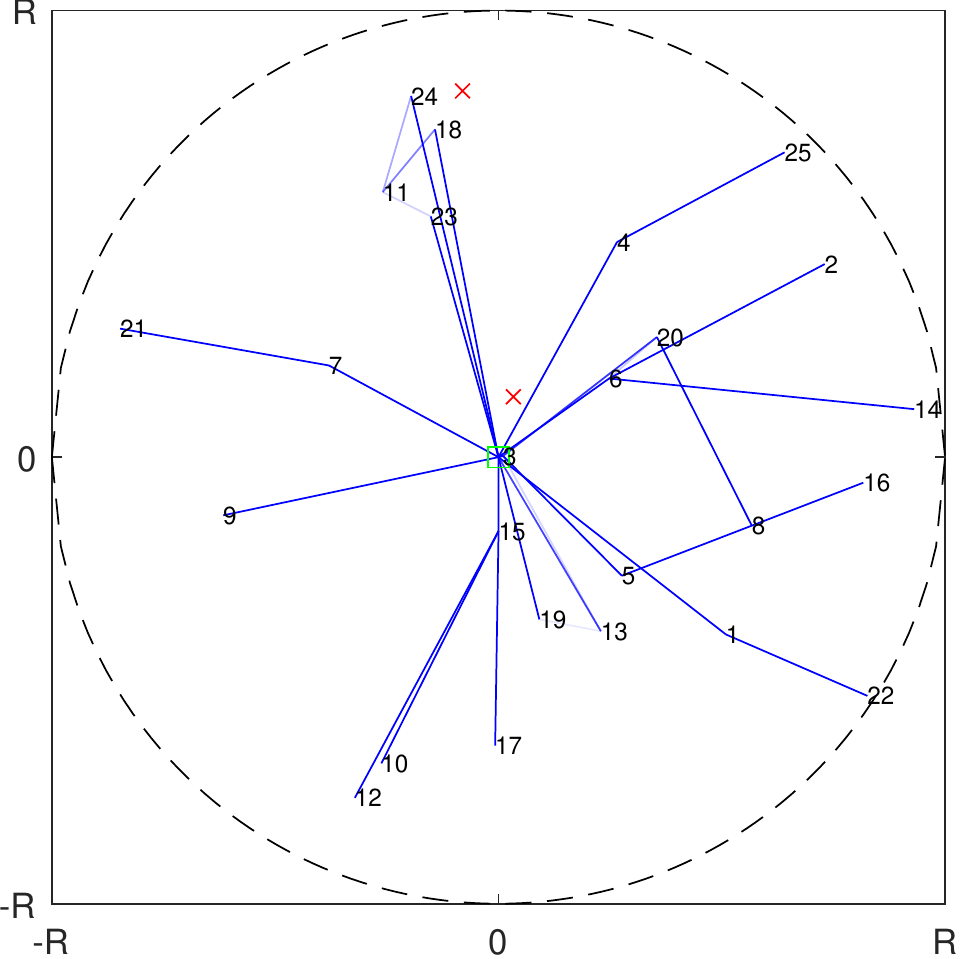} 
 }
    }
	\caption{Raw topology result by epoch on the $N_{d}=25$, $N_{b}=2$ example.}
	\label{Work1}
  \vspace{-0.5 cm}
\end{figure}

\begin{figure}[t]
	\centering
	\includegraphics[width=1.0\linewidth]{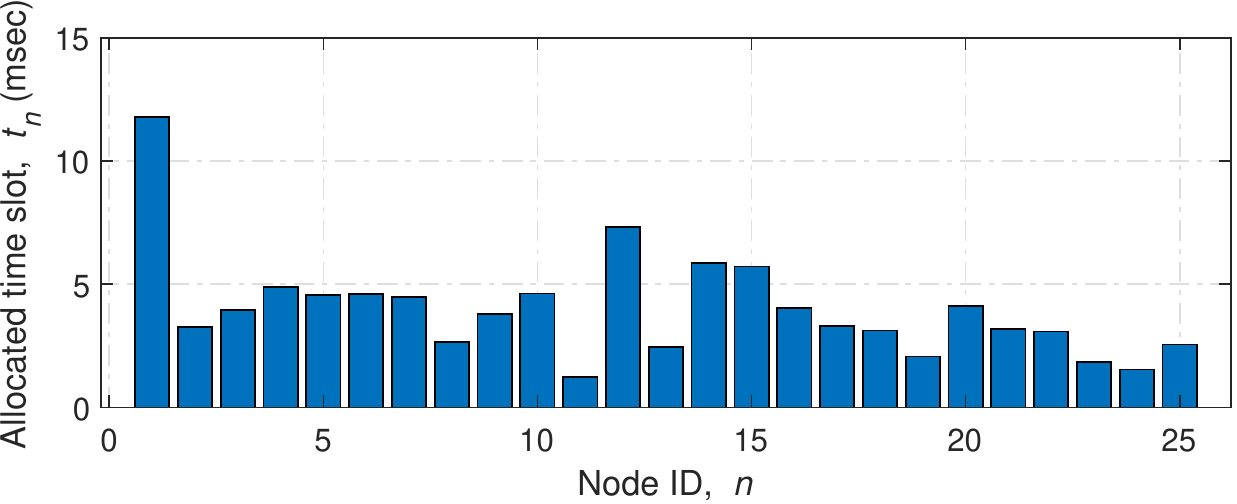} 
\vspace{-0.5cm}
	\caption{Allocated time slot on the $N_{d}=25$, $N_{b}=2$ example.}
	\label{TS1}
  \vspace{-0.5 cm}
\end{figure}

\begin{figure}[t]
	\centering
	\includegraphics[width=1.0\linewidth]{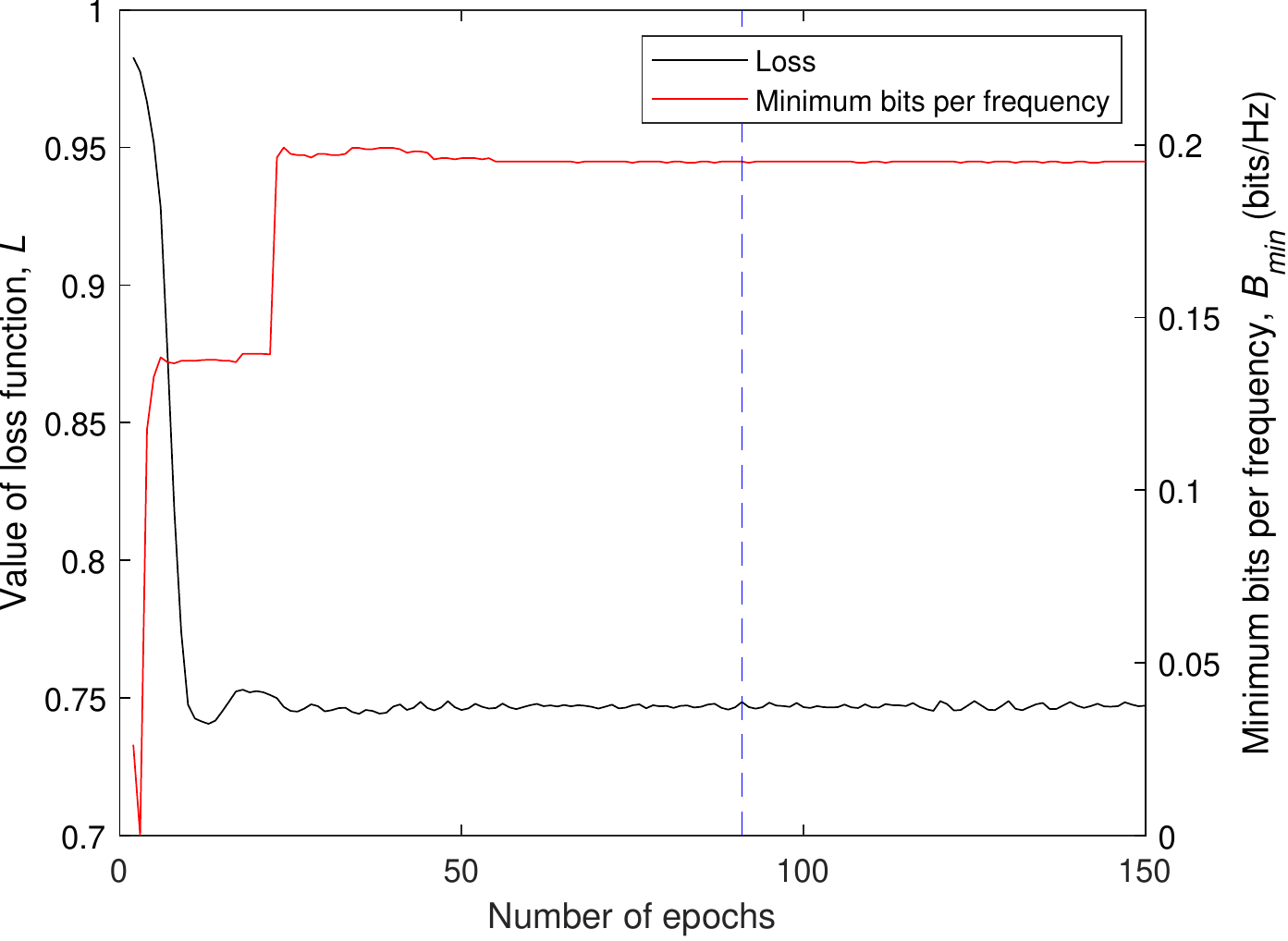} 
\vspace{-0.5cm}
	\caption{Training loss and minimum bits per frequency plot on the $N_{d}=25$, $N_{b}=2$ example.}
	\label{L1}
  \vspace{-0.5 cm}
\end{figure}

\section{Performance Analysis and Discussions}
\subsection{Test Environment and Parameter Configuration}

To evaluate the performance of our proposed method, we performed and evaluated simulations using the following target system model parameters.

First, the path loss exponent was assumed to be $\alpha=3$ and the bandwidth to be $BW = 125$ in kHz.
Noise power was defined as $N = -174 + NF + 10\log_{10}(BW)$ in dBm, where $NF$ is the noise figure equal to $6$ dB.
The nodes and PBs were distributed around a single packet sink with a coverage radius of $R = 0.5$ km.
Next, for the generation of wireless fading channels, $|h_{i,j}|^2$, an exponential random variable with unit mean was used.
The transmit power of the PB was set to 1 W.
The linear model was considered with $\eta=0.7$ for the energy harvesting model.
Finally, the time frame $T$ was set to $100$ milliseconds.

For our simulation, we set our VAE hyper-parameters as follows.
First, the learning rate was set to $\kappa = 0.001$, the decay rate for the first moment was set to $\beta_1 = 0.9$, the decay rate for the second raw moment was set to $\beta_2 = 0.999$, and the divide by zero prevention small scalar was set to $\epsilon = 10^{-8}$. 
Next, we set our VAE input vector size to $\left \lceil N_d\sqrt{N_d} \right \rceil$ and initialized our neural network weight using a glorot initializer \cite{Glorot10}.
Finally, we set our neural network bias to be zero initialized. 
We trained our VAE neural network until a training loss value smaller than the previous minimum training loss value was not found for $30+500/{N_d}$ epochs.

In addition, we set PT-EVM hyper parameters as follows: time slot scale factor $t = \frac{T}{N_d}$ and rate budget threshold $B_{\text{th}} = 0.01$.
Also, for the iterative balancing time slot allocation, the hyper-parameters were set as follows: the upper bound of the difference between $B_{\max}$ and $B_{\min}$ was $\epsilon_1 = 10^{-6}$bits/Hz and the minimum allocatable time slot was $\epsilon_2 = 10^{-7}$sec.

For a performance comparison, we adopted 3 other conventional schemes and optimal solution cases along with our proposed scheme:

\begin{figure}
	\centering
 \mbox{
 \subfigure[Proposed, 0.3 W]{\label{Ex5_0.3_Pro}
	\includegraphics[width=0.3\linewidth]{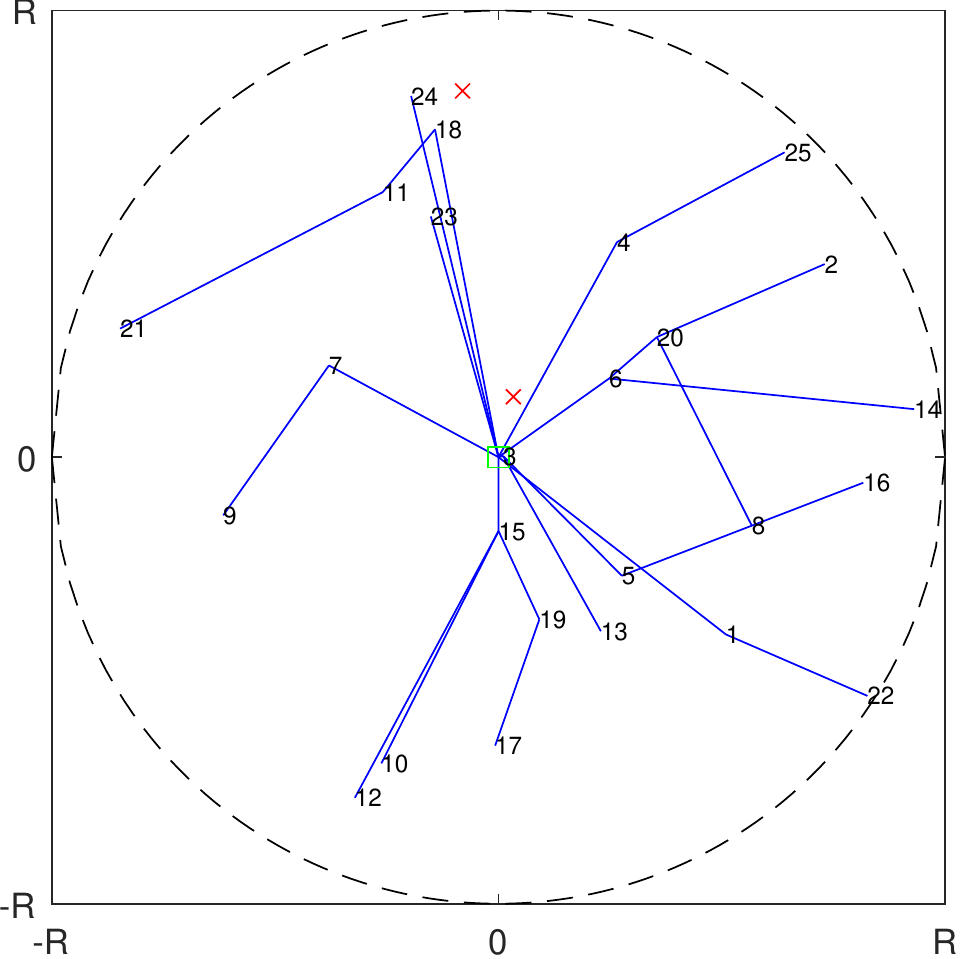} 
 }
  \subfigure[Proposed, 1 W]{\label{Ex5_1_Pro}
 \includegraphics[width=0.3\linewidth]{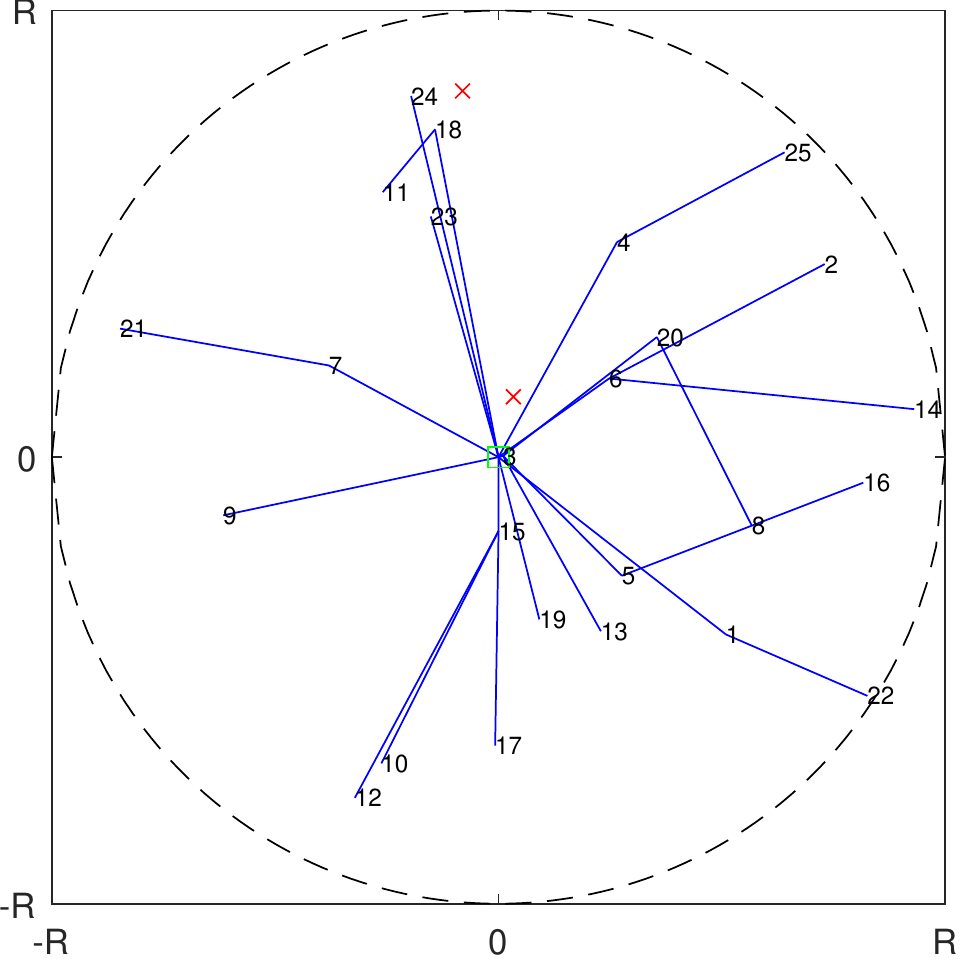} 
 }
  \subfigure[Proposed, 3 W]{\label{Ex5_3_Pro}
 \includegraphics[width=0.3\linewidth]{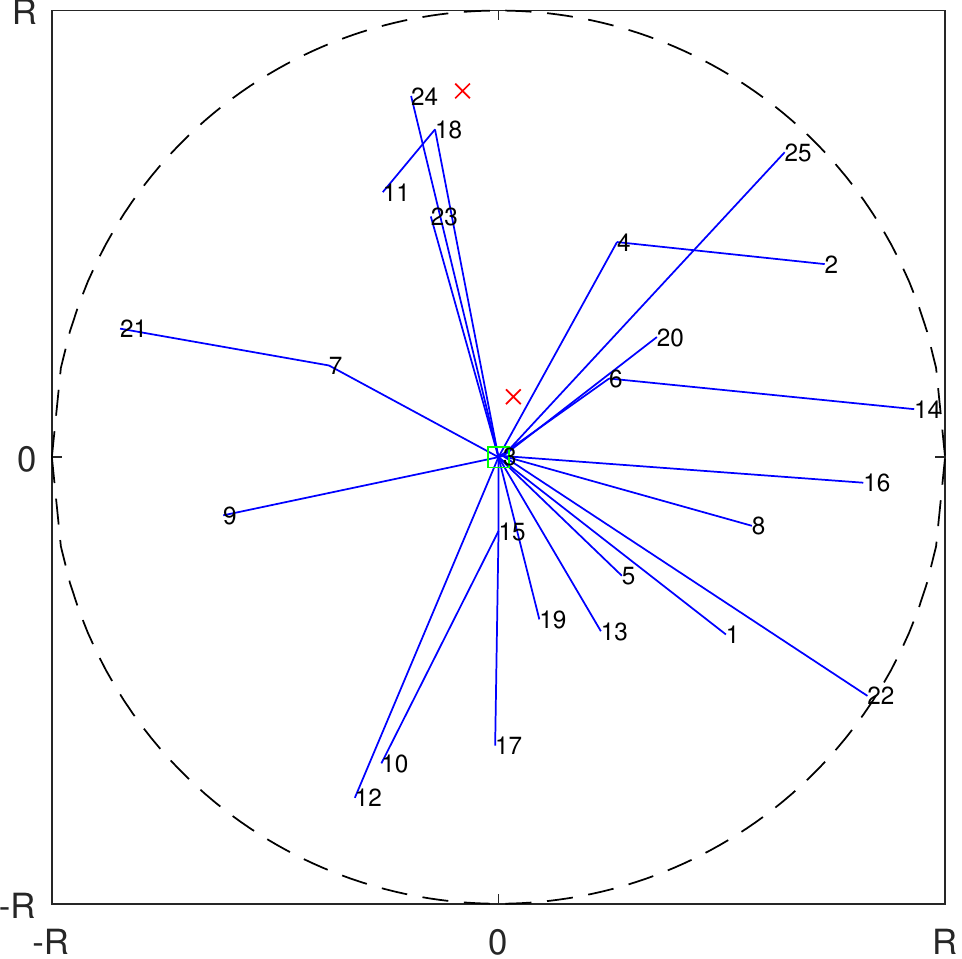} 
 }
    }
     \mbox{
 \subfigure[MST, 0.3 W]{\label{Ex5_0.3_MST}
	\includegraphics[width=0.3\linewidth]{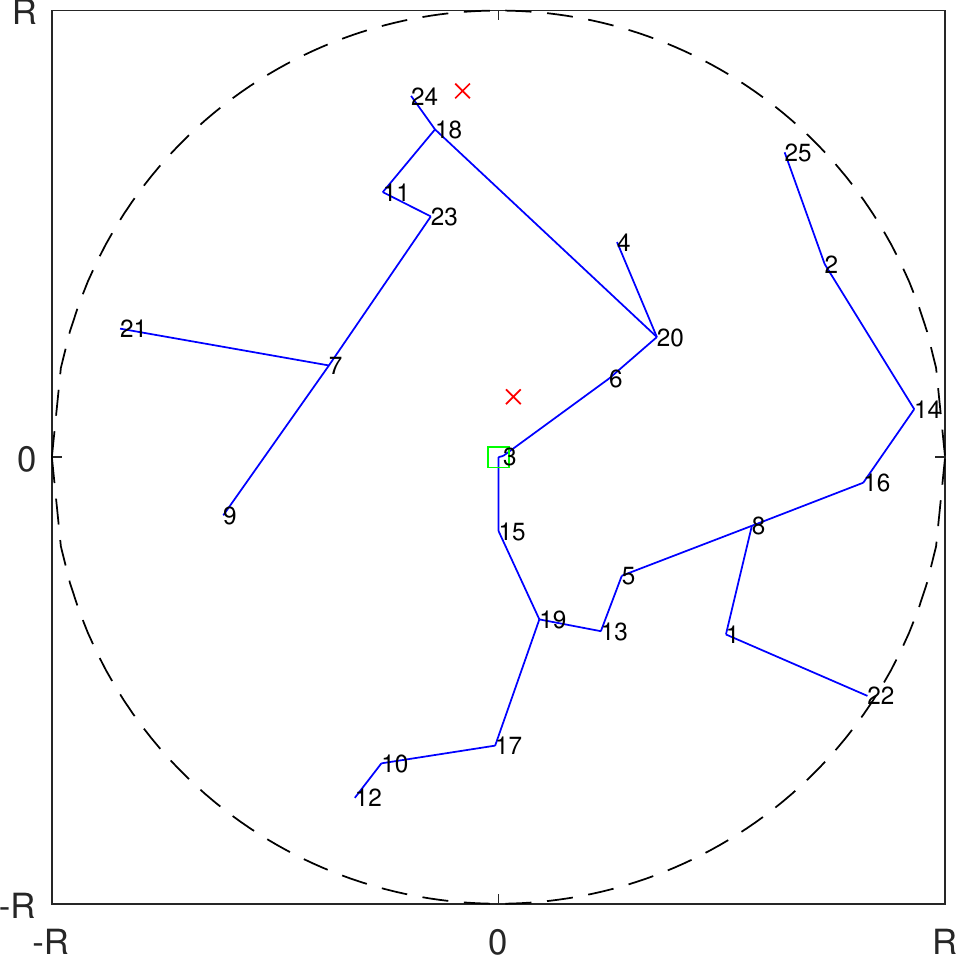} 
 }
  \subfigure[MST, 1 W]{\label{Ex5_1_MST}
 \includegraphics[width=0.3\linewidth]{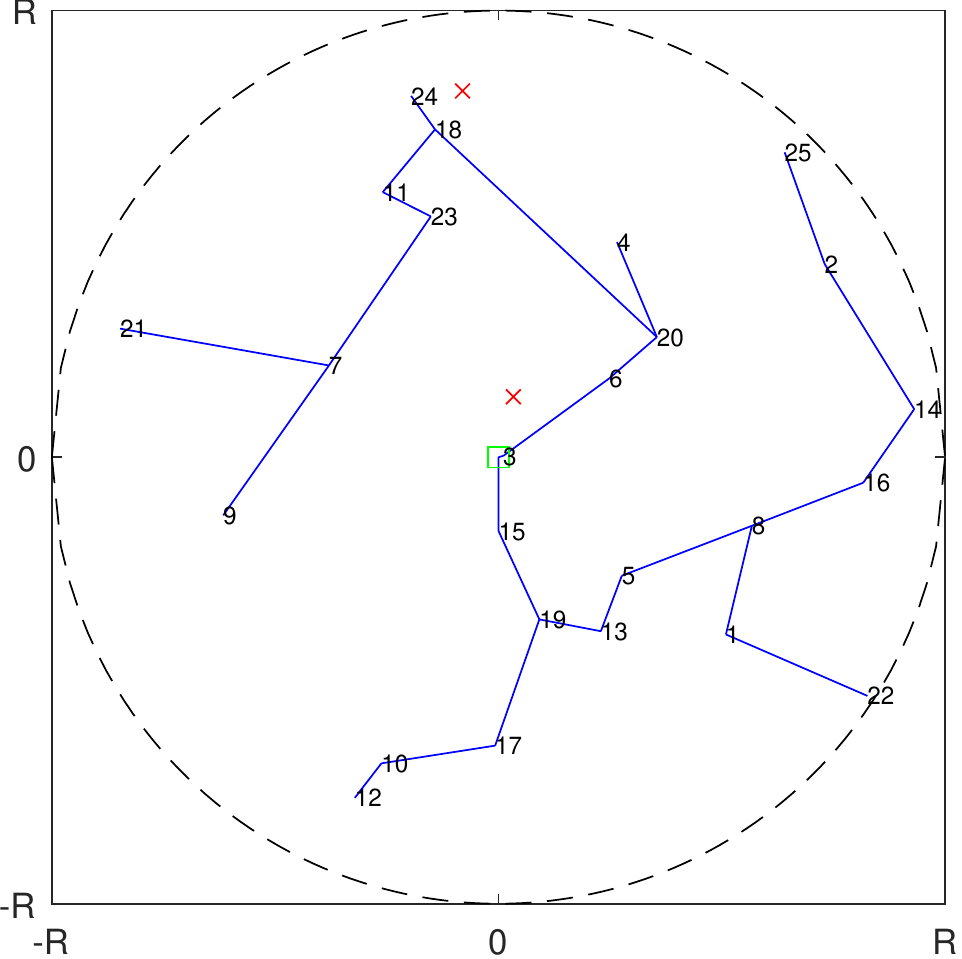} 
 }
  \subfigure[MST, 3 W]{\label{Ex5_3_MST}
 \includegraphics[width=0.3\linewidth]{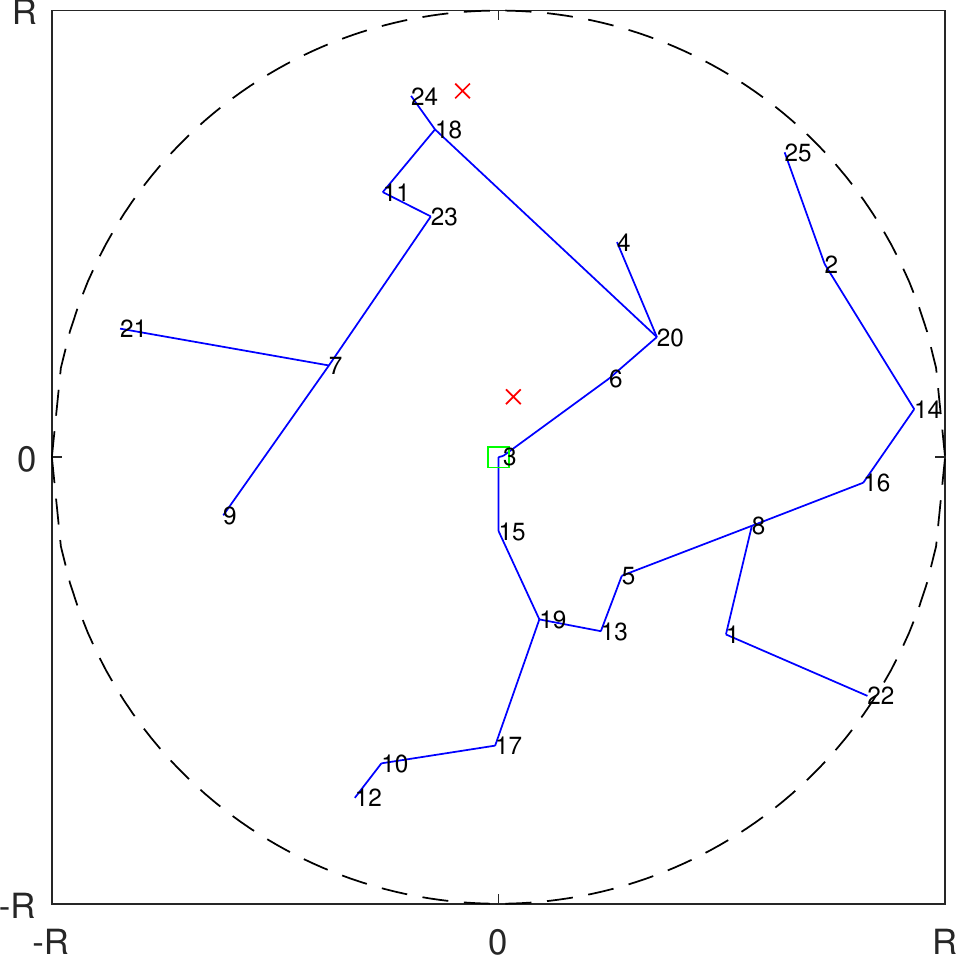} 
 }
    }
     \mbox{
 \subfigure[Greedy, 0.3 W]{\label{Ex5_0.3_Greedy}
	\includegraphics[width=0.3\linewidth]{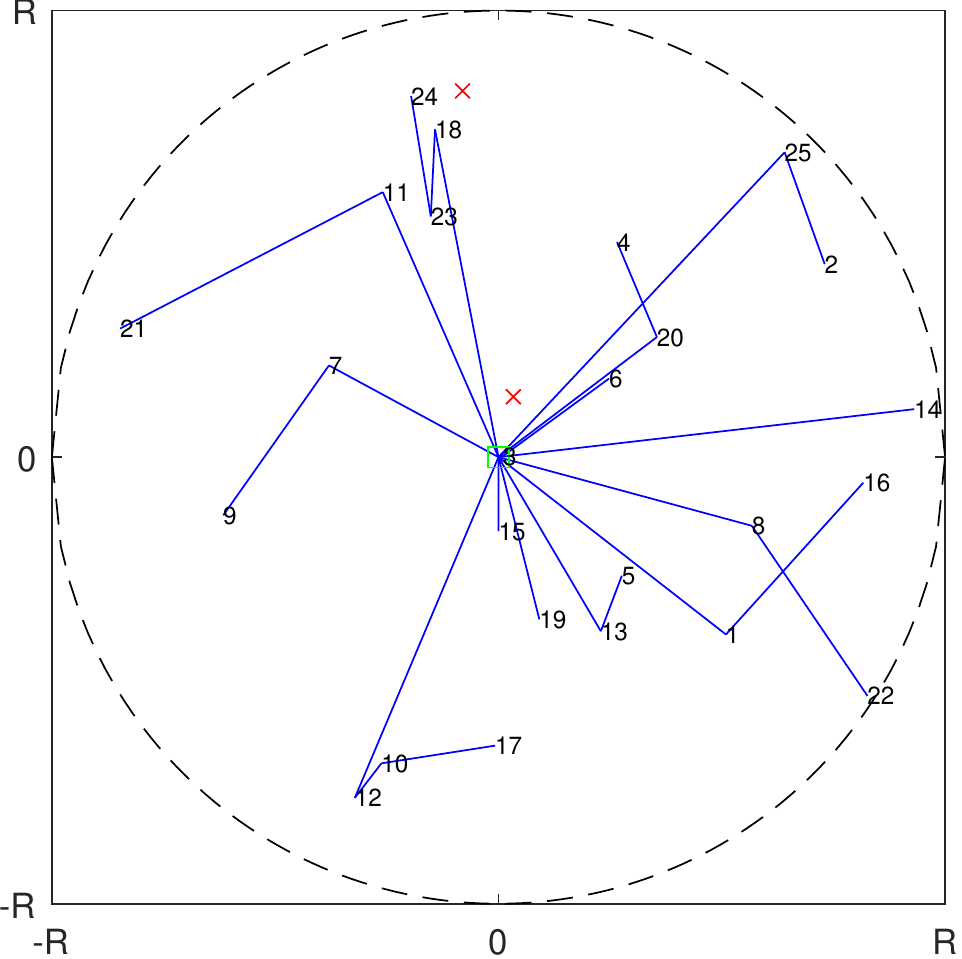} 
 }
  \subfigure[Greedy, 1 W]{\label{Ex5_1_Greedy}
 \includegraphics[width=0.3\linewidth]{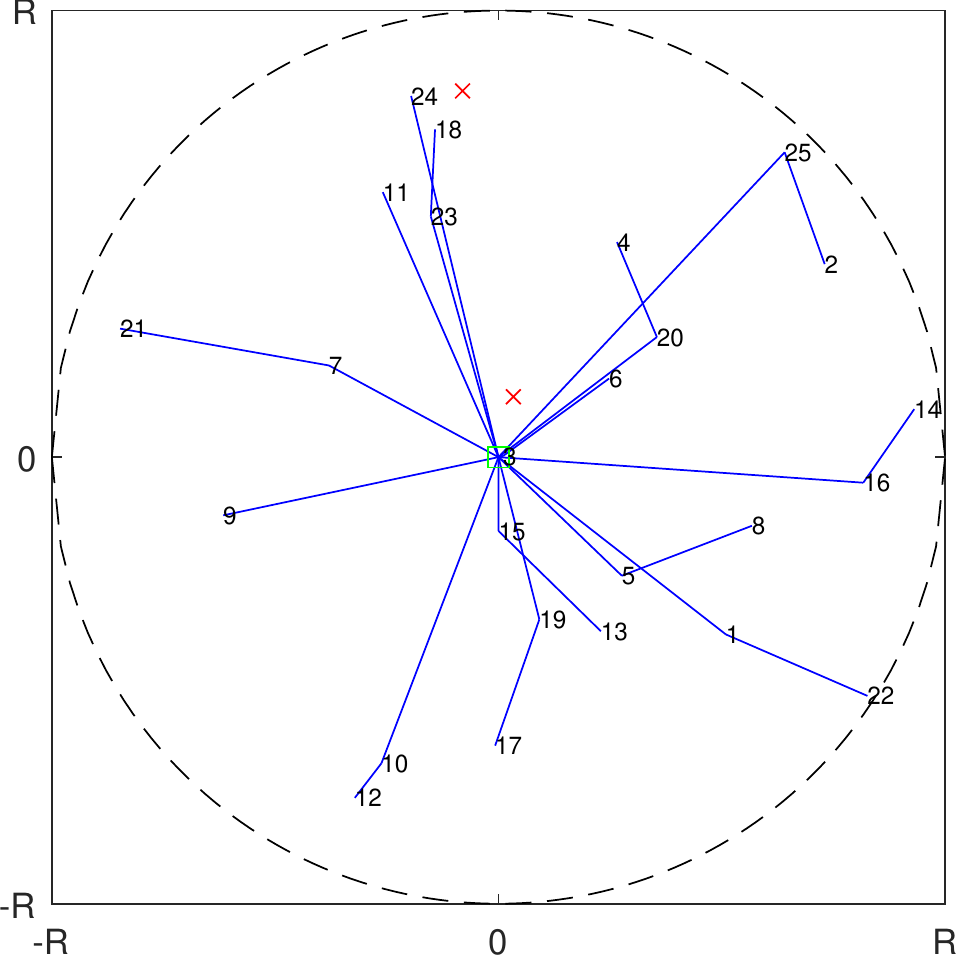} 
 }
  \subfigure[Greedy, 3 W]{\label{Ex5_3_Greedy}
 \includegraphics[width=0.3\linewidth]{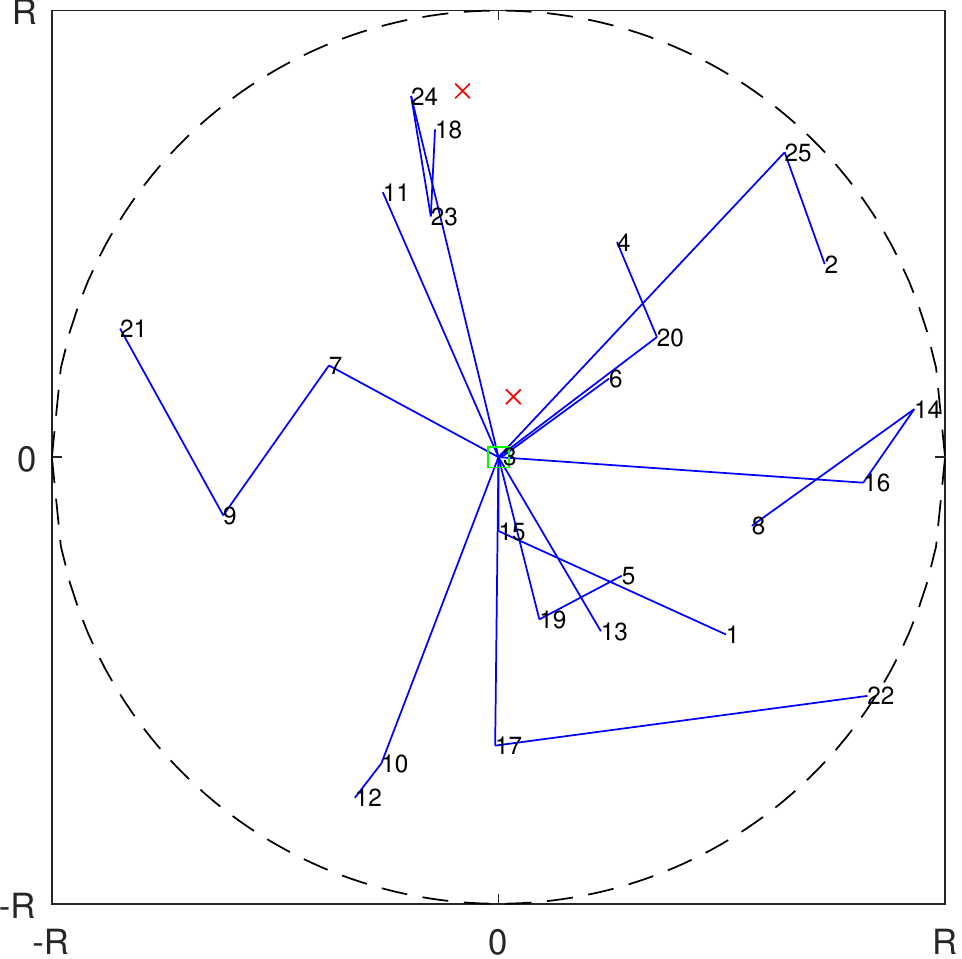} 
 }
    }
	\caption{Topology results over considered schemes with respect to the transmit power of the PB.}
	\label{Ex5}
\end{figure}

\begin{itemize}
    \item \textbf{Optimal solution (Opt)} which exhaustively searches all cases of valid topology configuration while allocating the time slot through Algorithm \ref{alg_TSalloc}, then chooses the best rate topology. 
    \item \textbf{Direct connect scheme (Dir)} which connects all nodes directly to the sink ($c_{n,N_d+1}=1$, $\forall n \in \mathbb{N}_d$), and allocates the time slot through Algorithm \ref{alg_TSalloc}.
    \item \textbf{MST scheme (MST)} which makes a Minimum Spanning Tree (MST) starting from the sink under initial $\mathbf{t}$ where $t_i$ = $T/N_d$ for $i \in \mathbb{N}_d$, based on link costs which are set to be reciprocal of the capacity, denoted as $1/\left( t_i \log _{2} (1+\Gamma_{i,j})\right)$, and allocates the time slot through Algorithm \ref{alg_TSalloc}.
    \item \textbf{Greedy scheme (Greedy)} which connects all nodes directly to the sink ($c_{n,N_d+1}=1$ $\forall n \in \mathbb{N}_d$), and repeats randomly selecting a node $i$ and connecting it to the one with highest achievable rate, $\min(t_i \log _{2} (1+\Gamma_{i,j}), \mathbf{B}_{j})$ where $\forall j \in \mathbb{N}_d$, then allocates the time slot through Algorithm \ref{alg_TSalloc}. 
    \item \textbf{Proposed scheme (Prop)} which decides the topology with our proposed VAE and PT-EVM scheme and allocates the time slot with Algorithm \ref{alg_TSalloc}.
\end{itemize}

All considered schemes were implemented on Matlab R2022b, using a computer equipped with an Intel Core i7-11700 and 16 GB of memory, without GPU acceleration.

\begin{table}
\centering
\begin{tabular}{|l|c|c|c|}
\hline
 & 0.3W & 1W & 3W \\
\hline
Proposed & 0.12324 & 0.19514 & 0.24229 \\
\hline
MST & 0.077389 & 0.093258 & 0.10730 \\
\hline
Greedy & 0.029918 & 0.15273 & 0.18587 \\
\hline
\end{tabular}
\caption{Minimum bits per frequency(bits/Hz) over the considered schemes with respect to the transmit power of the PB.}
\label{Ex5_Table}
 \vspace{-1.0 cm}
\end{table}

\begin{figure}[t]
	\centering
	\includegraphics[width=0.9\linewidth]{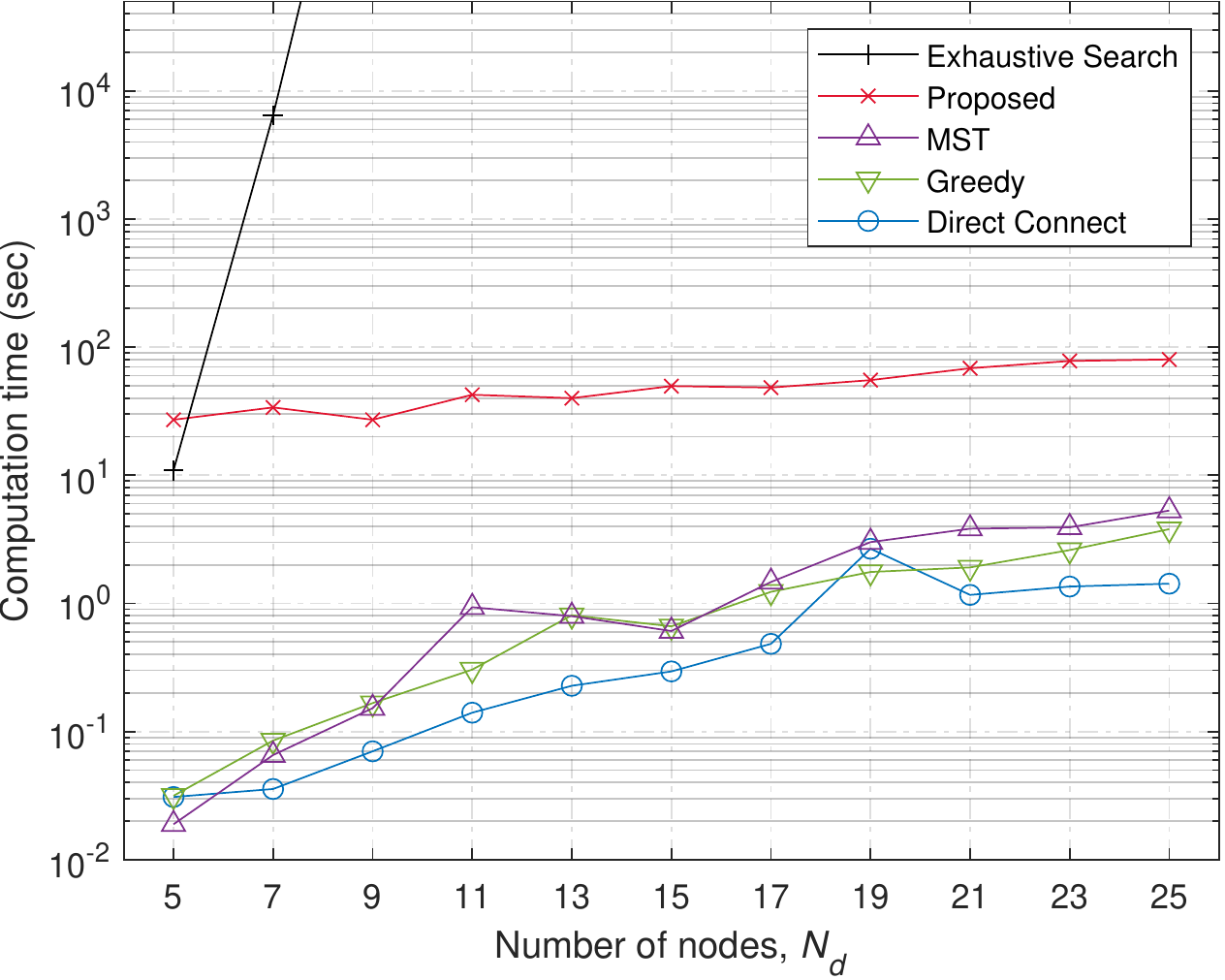} 
	\caption{Computation time over the considered schemes with respect to the number of nodes, $N_d$}
	\label{CT1}
  \vspace{-0.5cm}
\end{figure}

\subsection{Experiment Analysis and Discussion}

Fig. \ref{Work1} shows that the resulting raw topology (not post-processed) gradually changes in the example given in Fig. \ref{NetExample} as learning progresses.
Each sub-figure shows the raw topology result in epoch $0$, $10$, and $50$. We can see that our proposed scheme effectively converges to a sub-optimal topology in a reasonable training epoch. 
As shown in Fig. \ref{Work1}, the proposed VAE scheme converges to the final result topology within $50$ iterations (roughly taking $60$ seconds to compute). 
the enlarged final result topology for this network example is depicted in Appendix A.

Fig. \ref{TS1} shows the time slot allocation results of the example given in Fig. \ref{NetExample}. 
Note that relaying nodes such as node ID 1, 4, 5 and 15 are allocated relatively more time slots, as well as the nodes in an inferior power/position environment such as node ID 12 and 14.
Because of this, the proposed IB time slot allocation algorithm effectively allocates the time slot fairly.

Fig. \ref{L1} shows the training loss and performance (bits/Hz, $B_{\min}$) convergence with respect to the number of training epochs for the proposed VAE and PT-EVM hybrid scheme under the $N_d = 25, N_b = 2$ sample, previously described in Fig. \ref{NetExample}.

\begin{figure*}
	\centering
	\mbox{\subfigure[$N_d=5$]{\label{Ex1-1} \includegraphics[width=0.33\linewidth]{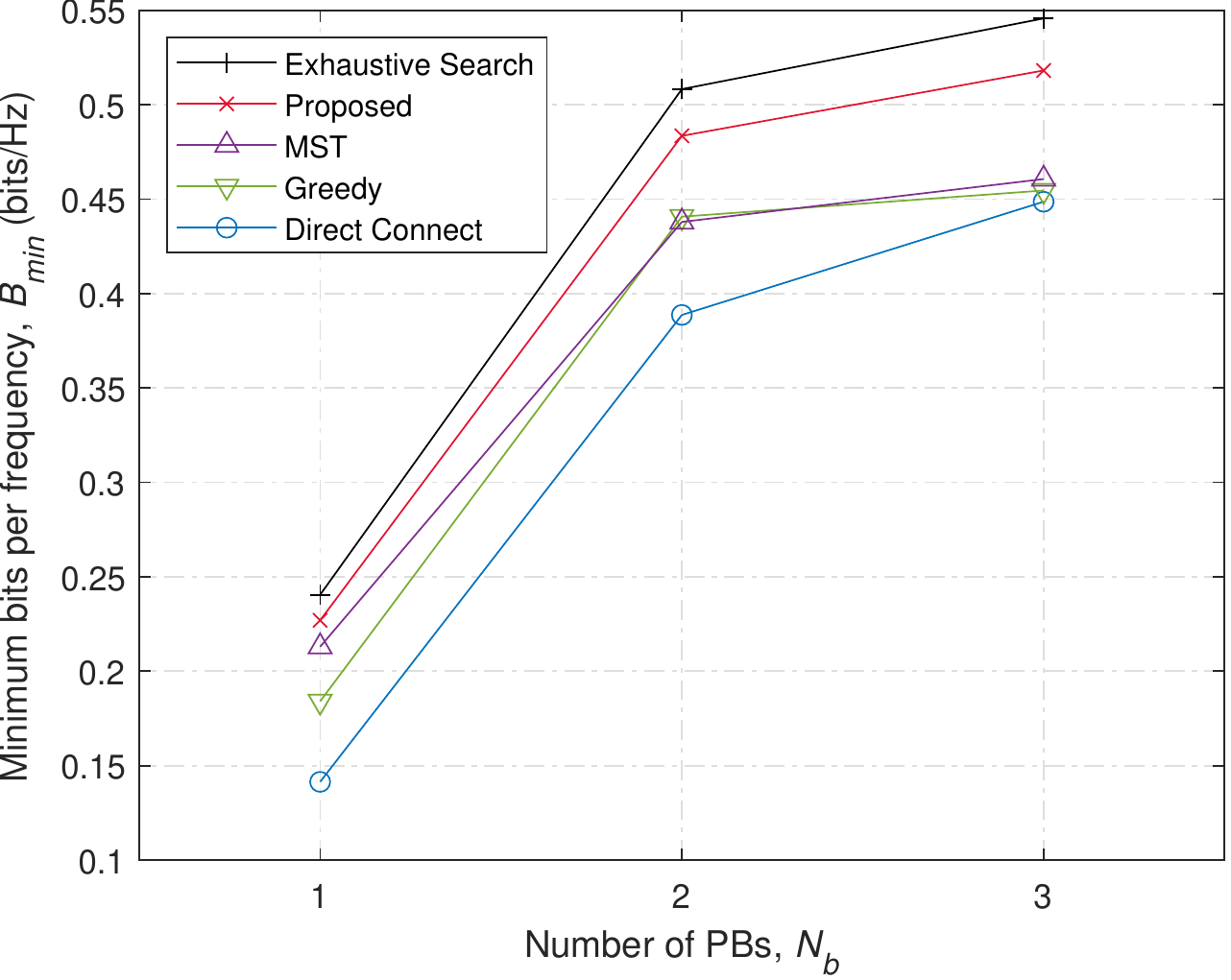}
	}
	\subfigure[$N_d=6$]{\label{Ex1-2} \includegraphics[width=0.33\linewidth]{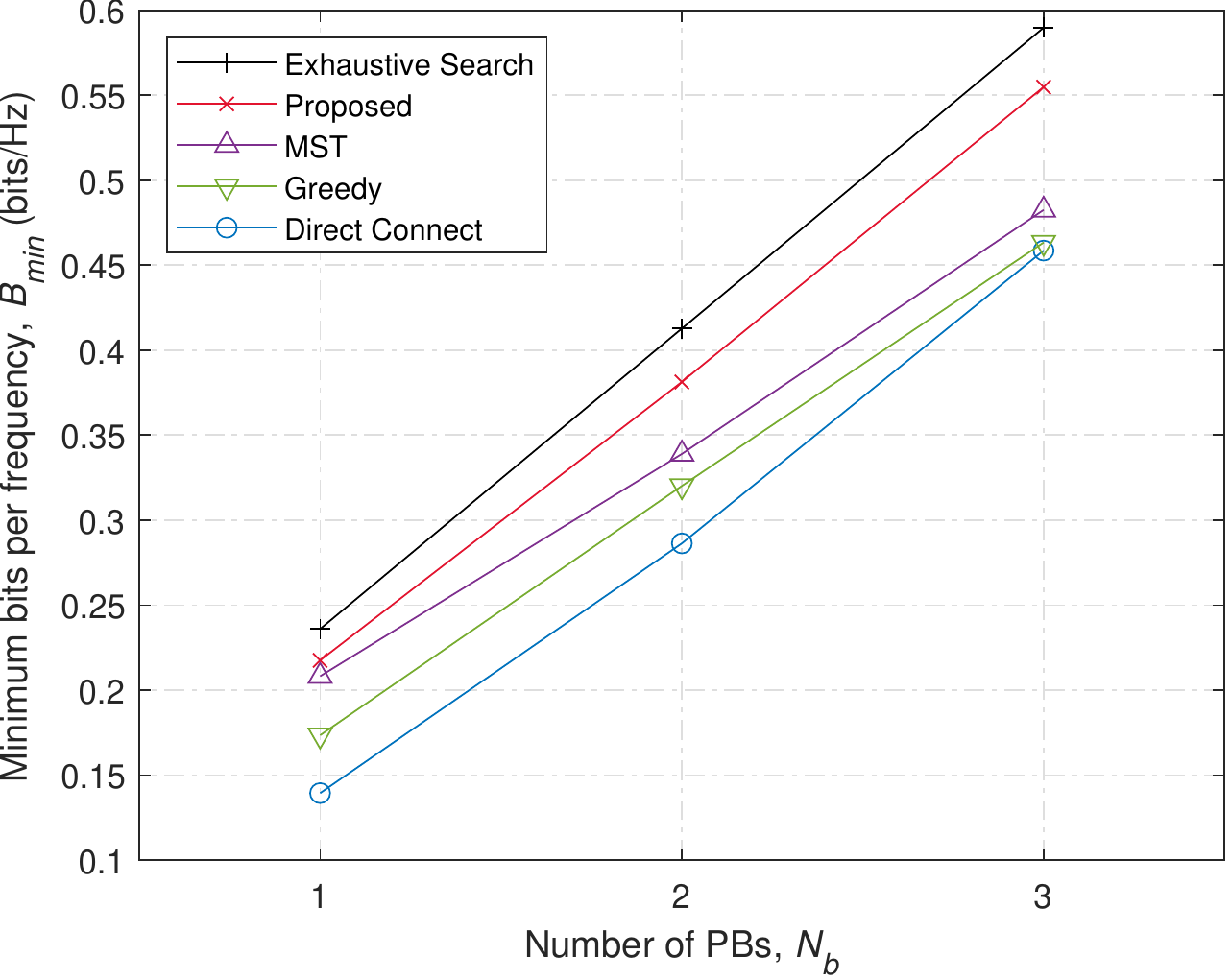}
    }
    \subfigure[$N_d=7$]{\label{Ex1-3} \includegraphics[width=0.33\linewidth]{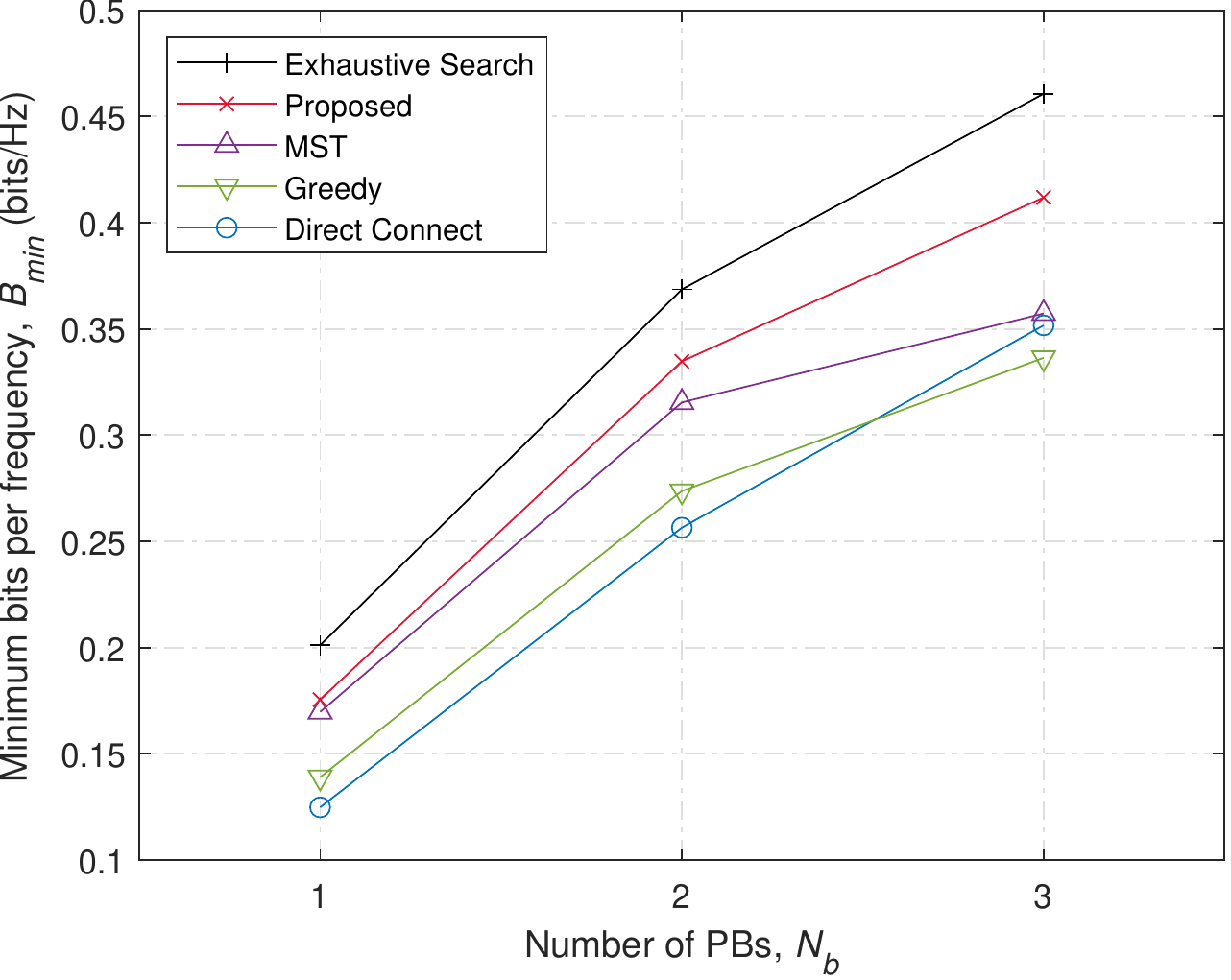}
	}
    }
	\caption{\mbox{Performance comparison with respect to the number of PBs, $N_b$ in $N_d\in\{5,6,7\}$}}
	\label{Ex1}
  \vspace{-0.3cm}
\end{figure*}

\begin{figure*}
	\centering
	\mbox{\subfigure[$N_d=10$]{\label{Ex2-1} \includegraphics[width=0.33\linewidth]{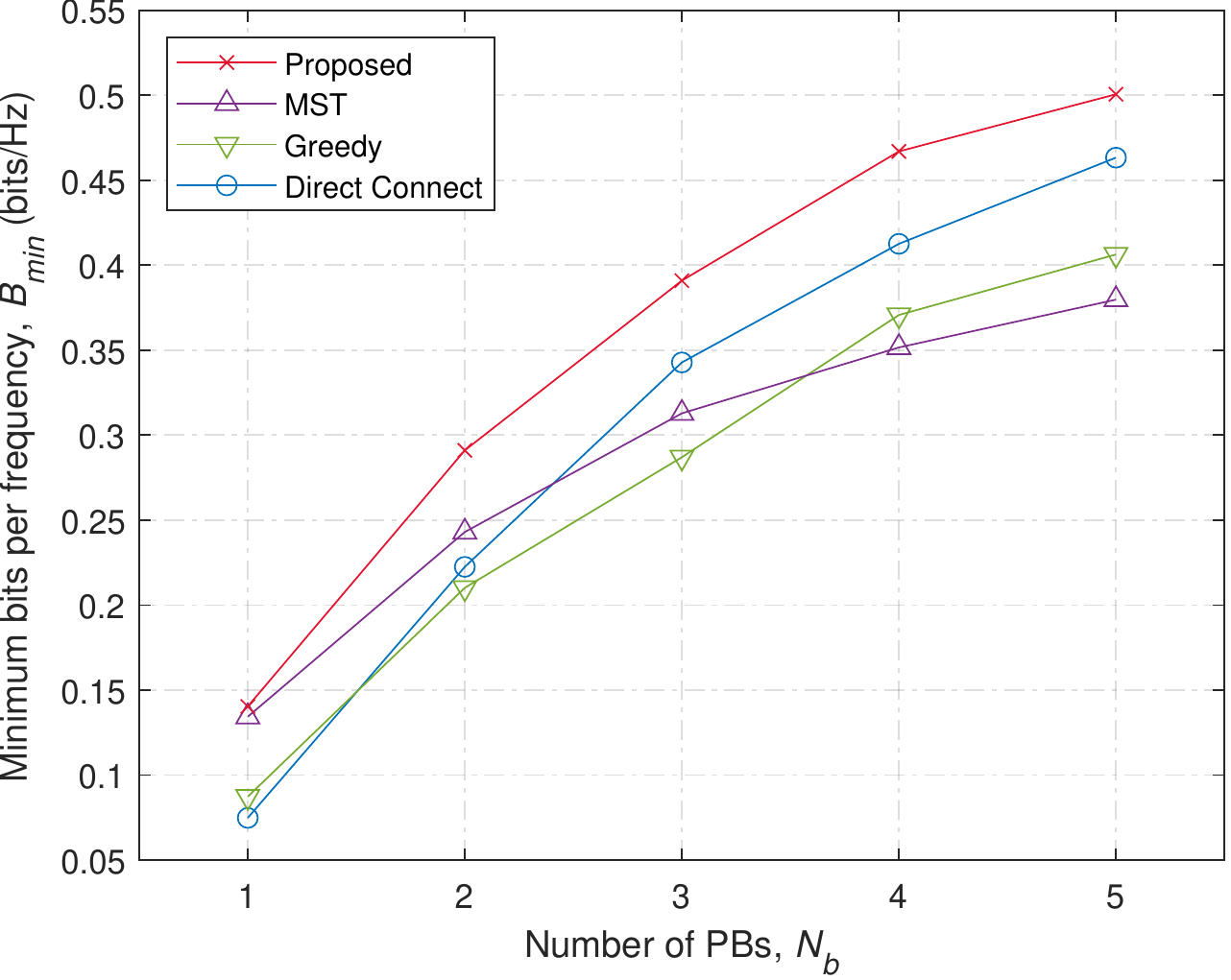}
	}
	\subfigure[$N_d=20$]{\label{Ex2-2} \includegraphics[width=0.33\linewidth]{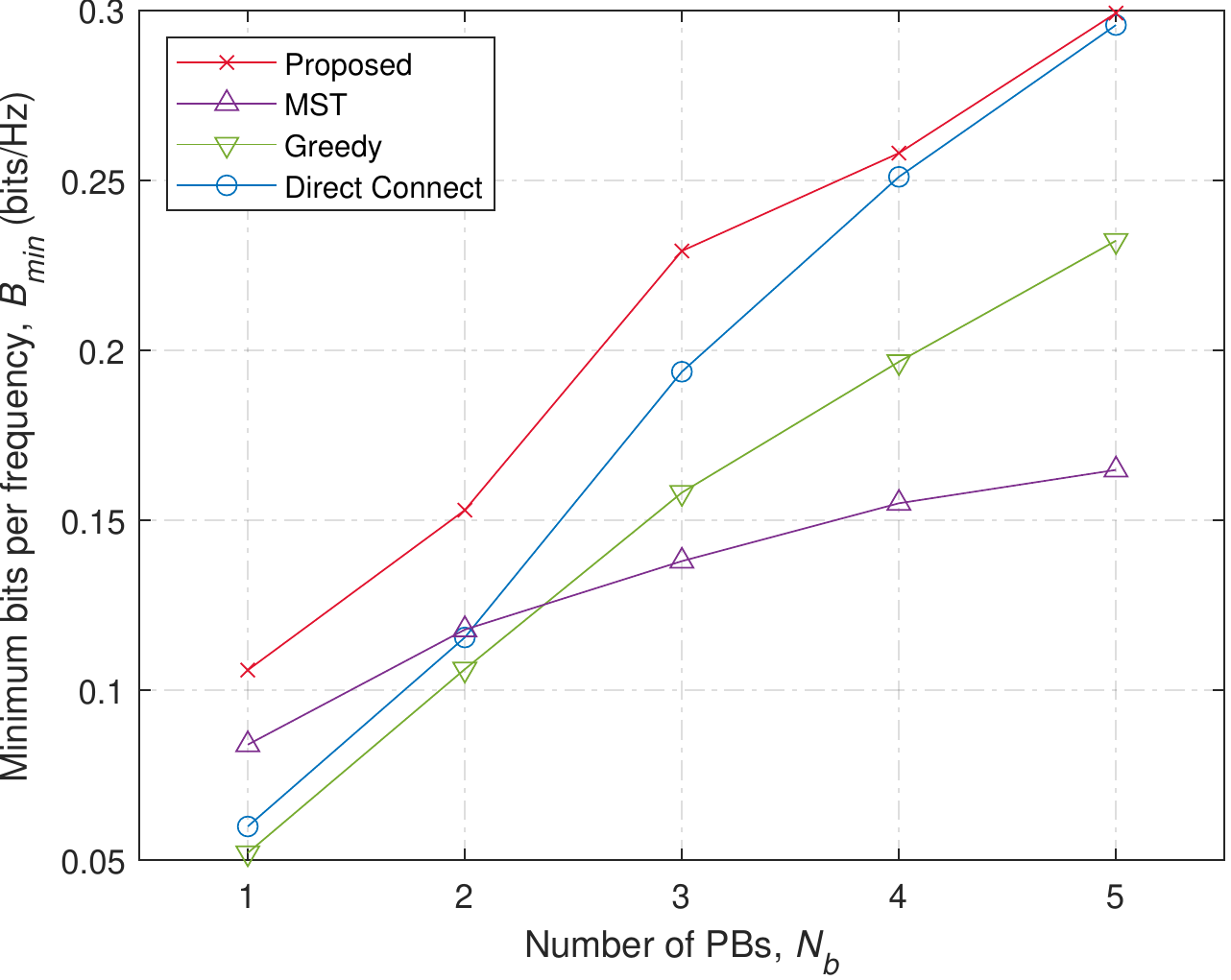}
    }
    \subfigure[$N_d=30$]{\label{Ex2-3} \includegraphics[width=0.33\linewidth]{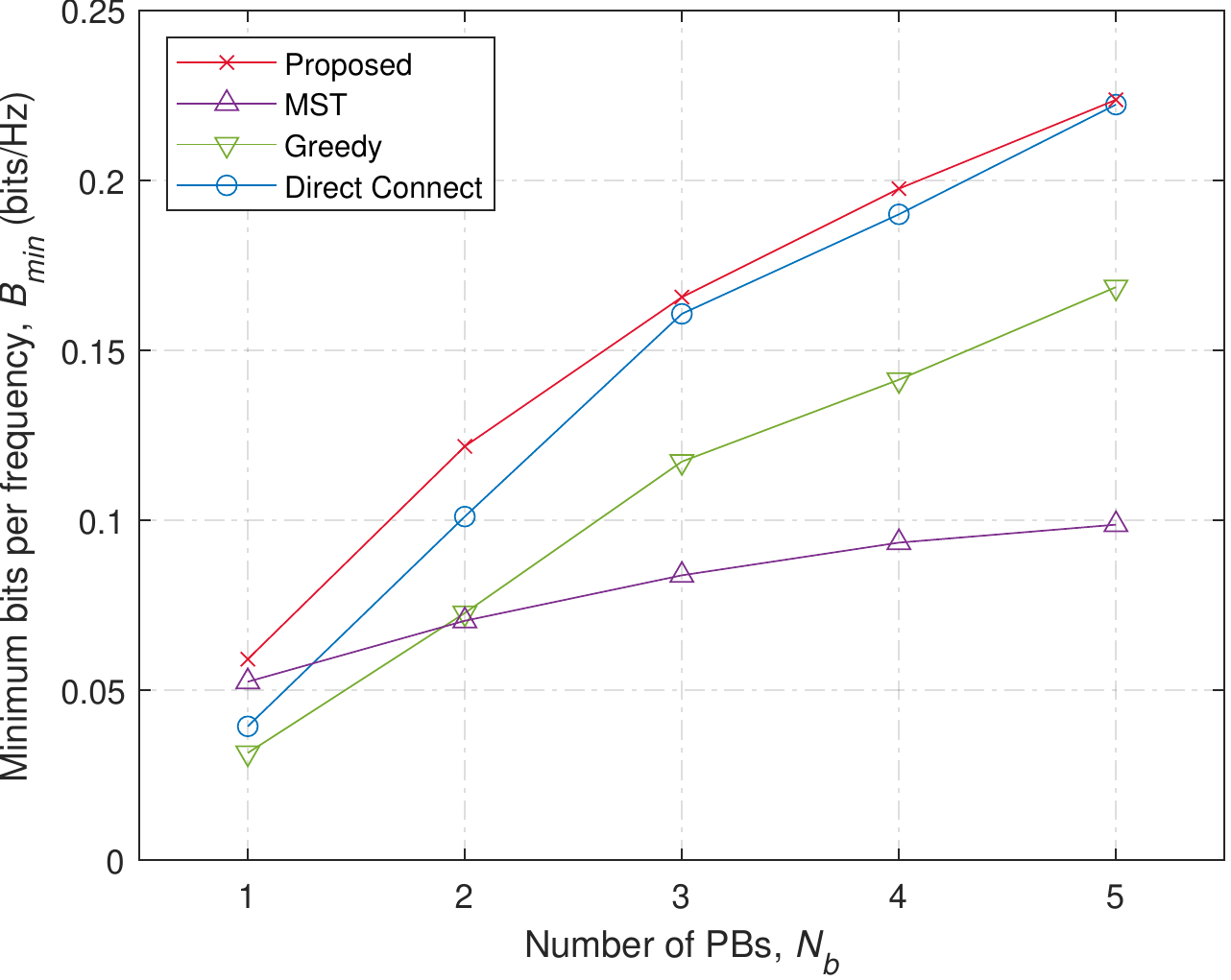}
	}
    }
	\caption{\mbox{Performance comparison with respect to the number of PBs, $N_b$ in $N_d\in\{10,20,30\}$}}
	\label{Ex2}
  \vspace{-0.3cm}
\end{figure*}

The blue dashed line in Fig. \ref{L1} indicates the point at which training has ended. 
The proposed VAE scheme ends its training loop at epoch 91. 
However, we showed the result of the plot over the end of the training loop so that we can observe what happens when training loop continues after the expected learning termination point.
It can be confirmed that our proposed scheme is both stable and non-overfitting since training loss and performance remain constant after the learning progresses beyond the termination point of the learning epoch.

    The training loss function value is defined as $L = \frac{1}{N_d} \sum_{i \in \mathbb{N}_d}{\exp(-R_{\text{sim}})_i}$ and the value $(R_{\text{sim}})_i$ has a very small metric due to the previously described characteristics of IoT networks.
    For example, Fig. \ref{L1} shows that $(R_{\text{sim}})_i$ have an initial value of zero and a final value of about 0.2 for the network configuration in Fig. \ref{NetExample}, meaning that the training loss value $L$ starts from $1$ and finalizes at approximately 0.7 to 0.8, respectively. 
    Therefore, the training loss value is bound to show minimal change.

Fig. \ref{Ex5} and Table \ref{Ex5_Table} show the topology results and performances (bits/Hz) over the considered schemes with respect to the transmit power of the PB.
It can be seen that the topology results by the proposed scheme adaptively change depending on the transmit power of the PB, showing a preference for direct connections in a high power environment and a preference for relaying connections in a low power environment.
The MST algorithm results in identical topology despite the changing power configuration, showing a lack of adaptability for various IoT network configurations.
As shown in Table \ref{Ex5_Table}, we can also confirm that our proposed scheme outperforms other schemes over different configurations of the transmit power of the PB.

Fig. \ref{CT1} shows the measured computation time among the considered schemes with respect to the number of nodes from $N_d = 5$ to $25$ with an interval of $2$. 
Note that the number of PBs does not affect the measured computation time.
We can see that the computation time of the optimal solution search grows exponentially, approximately 2 hours for the $N_d=7$ case, which makes the exhaustive search method infeasible in practical environments, e.g., $N_d \geq 8$. 
Consequently, the computation time of the optimal solution is not suggested over $N_d \geq 8$. 
This suggests the need of an alternative scheme with sub-optimal performance and rational computation time. 

Although our proposed VAE and PT-EVM scheme takes more computation time over other considered schemes, the computation time gap between the proposed scheme and conventional schemes decreases as the order of magnitude of $N_d$ increases.
This makes the proposed scheme plausible for use over large IoT networks.

Fig. \ref{Ex1-1} to \ref{Ex1-3} shows the minimum bits per frequency with respect to $N_b \in \{ 1,\,2,\,3\}$ and $N_d\in \{5,\,6,\,7\}$, respectively, over all the considered schemes. 
Each data point on Fig. \ref{Ex1-1} to \ref{Ex1-3} is averaged over the performance results from $30$ random distributions of nodes and PBs.
The minimum bits/Hz could be converted to a rate by multiplying $\frac{BW}{T}$, and the difference between the minimum and maximum bits/Hz is guaranteed to be below $10^{-6}$ since Alg. \ref{alg_TSalloc} is implemented with $\epsilon_1 = 10^{-6}$ bits/Hz.
Fig. \ref{Ex1-1} to \ref{Ex1-3} shows that the minimum bits/Hz increases over all considered schemes as $N_b$ increases since the received power at each node increases. 
In particular, the MST scheme shows relatively good performance in the low power condition, while the direct connect scheme shows good performance in the high power condition; therefore, those schemes have strengths in different power environments.
Conversely, our proposed scheme achieves better performance than all other considered schemes at all configurations of number of nodes and PBs: closest to the optimal solution we found.

Fig. \ref{Ex2-1} to \ref{Ex2-3} shows the minimum bits/Hz with respect to $N_b\in \{1,\,2,\,3,\,4,\,5\}$ and $N_d\in\{10,\,20,\,30\}$. Note that each data point on Fig. \ref{Ex2-1} to \ref{Ex2-3} is also averaged over the performance results from $30$ random distributions of nodes and PBs.
the greedy method shows lower performance than the direct connect method compared to fewer $N_d$s, as shown in Fig. \ref{Ex1-1} to \ref{Ex1-3}, since the optimal choice from a local perspective may not be appropriate from a global perspective, thus deteriorating the overall performance of the network.
Our proposed scheme is also confirmed to be superior over the three other considered schemes, ranked at the highest performance in the graph at all simulation configurations.
In these cases, we do not provide optimal solutions due to the excessive computation time of an exhaustive search.
Also, as shown in Fig. \ref{Ex1-1} to \ref{Ex2-3}, our proposed scheme is superior among various $N_d$ from 5 to 30, illustrating the scalability of our method.

Therefore, we can confirm that our proposed scheme outperforms other existing schemes while achieving sub-optimal performance.
Subsequently, we can infer three reasons for the performance superiority and adaptability towards varying environment configurations of our proposed scheme. 
i) The VAE stage has a non-linearity property within the model, thereby effectively coping with our NP-hard mixed integer system model problem.
 ii) The backward-pass based evaluation stage and its novel packet-tracing algorithm gives a sufficiently concise and accurate assessment for output topology from the VAE part result simultaneously, thus providing the VAE stage with an appropriate learning direction. 
 iii) Lastly, our proposed IB based time allocation algorithm fairly distributes the time slot in the TDMA system, resulting in effective final fine-tuning on the formulated max-min problem.

\section{Conclusions}

In this study, we formulated a max-min optimization problem for the TDMA based IoT relay network with energy harvesting (EH).
We proposed a VAE based module to effectively solve the formulated problem, in spite of a lack of dataset and the non-linearity characteristic of problem. 
We also proposed a novel backward-pass based assessment algorithm called "Packet-Tracing" to precisely and concisely assess and train our proposed VAE module. 
Finally, we proposed an IB time slot allocation algorithm to achieve TDMA fine-tuning optimization in the fairness aspect.
We presented a practical example of our proposed scheme running along with the loss and performance plot, showing its stability and performance.
We also provided a computation and performance comparison with three other conventional schemes and the optimal brute-force solution.
We observed and confirmed that our proposed scheme is both stable and superior to other considered schemes through numerical simulations.

\begin{appendices}
\section{Output of the proposed scheme in the example given in Figure 1}
\begin{figure}[H]
	\centering
	\includegraphics[width=0.7\linewidth]{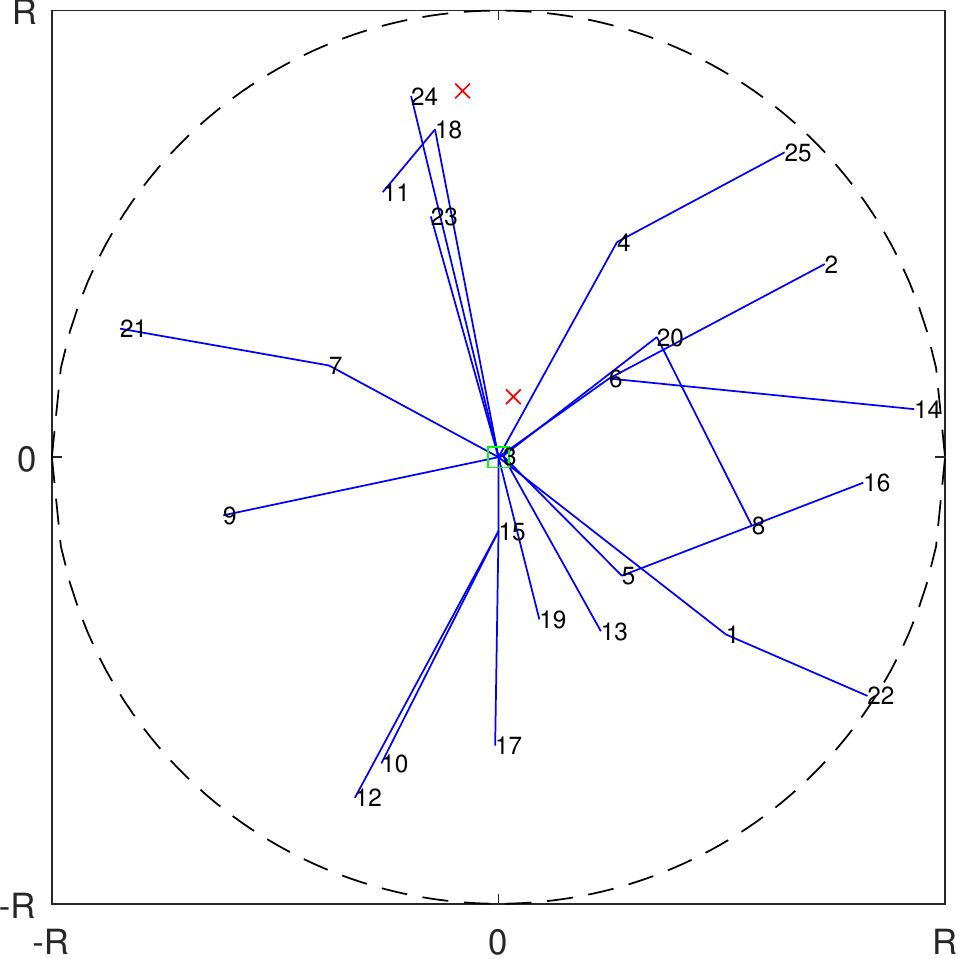} 
	\caption{Output of proposed scheme in the example given in Fig. \ref{NetExample}} 
	\label{NetExample_result}
\end{figure}

\end{appendices}


\bibliographystyle{IEEEtran}

\begin{IEEEbiography}[{\includegraphics[width=1in,height=1.25in,clip,keepaspectratio]{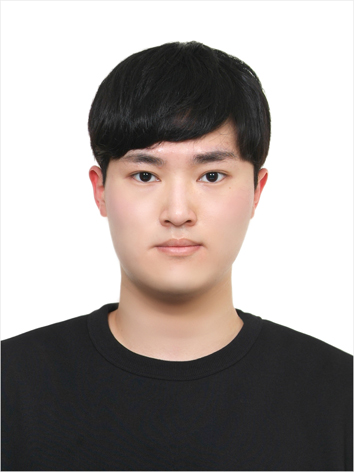}}]{Kiseop Chung} received a B.S degree in Electrical and Computer Engineering from Seoul National University, Seoul, South Korea, in 2022. Since June 2022, he has been a research officer at the Agency for Defense Development (ADD), South Korea. His research interests include internet of things (IoT), wireless networks, unsupervised machine learning, embedded systems, and hardware architecture.
\end{IEEEbiography}

\begin{IEEEbiography}[{\includegraphics[width=1in,height=1.25in,clip,keepaspectratio]{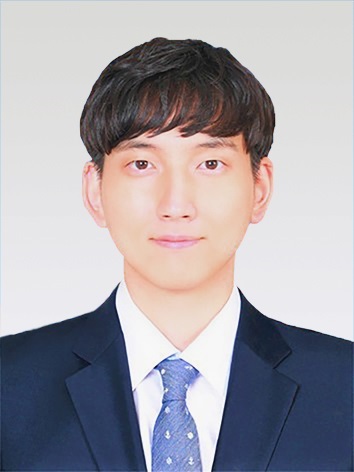}}]{Jin-Taek Lim} (S'14--M'19) received a B.S degree in Electrical and Electronic Engineering from Yonsei University, Seoul, South Korea, in 2012, and M.S. and Ph.D. degrees in Electrical Engineering from the Korea Advanced Institute of Science and Technology (KAIST), Daejeon, South Korea, in 2014 and 2019, respectively. From March 2019 to April 2023, he was a senior researcher at the Agency for Defense Development (ADD), South Korea. He currently joined Samsung Electronics, South Korea, in May 2023, as a senior researcher. His research interests include internet of things (IoT), simultaneous wireless information and power transfer (SWIPT), information security, and full-duplex systems.
\end{IEEEbiography}

\EOD

\end{document}